\documentclass[11pt]{article}
\pdfoutput=1
\usepackage[margin=1in]{geometry}
\usepackage[utf8]{inputenc}
\usepackage[T1]{fontenc}
\usepackage{microtype}
\usepackage{amsmath, amssymb, amsthm, dsfont}
\usepackage{enumitem}
\usepackage{graphicx, subfig}
\usepackage[dvipsnames]{xcolor}
\usepackage{hyperref, comment,url}
\usepackage[ruled,vlined,linesnumbered]{algorithm2e}
\usepackage{authblk}
\usepackage{tikz}
\usetikzlibrary{calc}
\usepackage{complexity}

\DeclareMathOperator{\ub}{ub}
\newcommand{\ubr}{r_{\ub}}

\allowdisplaybreaks

\usepackage{color}

\title{The Randomized Communication Complexity of Randomized Auctions}
\author{Aviad Rubinstein\thanks{Supported by NSF CCF-1954927, and a David and Lucile Packard Fellowship.}\\Stanford University\\\texttt{aviad@cs.stanford.edu} \and Junyao Zhao\thanks{Supported by NSF CCF-1954927.}\\ Stanford University\\\texttt{junyaoz@stanford.edu}}

\def \polylog {\textnormal{polylog}}
\def \C {\mathcal{C}}
\def \D {\mathcal{D}}

\def \F {\mathcal{F}}
\def \cA {\mathcal{A}}
\def \cB {\mathcal{B}}
\def \cM {\mathcal{M}}
\def \cR {\mathcal{R}}
\def \cT {\mathcal{T}}

\def \P {\textnormal{Pr}}
\def \R {\mathbb{R}}
\def \X {\mathcal{X}}
\def \eg {{e.g.}}

\DeclareMathOperator*{\Beta}{Beta}

\newtheorem{theorem}{Theorem}[section]
\newtheorem{thm}[theorem]{Theorem}
\newtheorem{lemma}[theorem]{Lemma}
\newtheorem{corollary}[theorem]{Corollary}

\newtheorem{prop}[theorem]{Proposition}
\newtheorem{claim}[theorem]{Claim}

\newtheorem{example}[theorem]{Example}

\theoremstyle{definition}
\newtheorem{definition}[theorem]{Definition}
\newtheorem{defn}[theorem]{Definition}
\newtheorem{remark}[theorem]{Remark}
\newtheorem{rem}[theorem]{Remark}
\newtheorem{OQ}[theorem]{Open Question}

\newtheorem*{thm*}{Theorem}

\ifx\mute\undefined
\newcommand{\aviad}[1]{{\color{purple} [[AVIAD: #1]]}}
\newcommand{\junyao}[1]{{\color{green} [[Junyao: #1]]}}

\else
\newcommand{\aviad}[1]{}
\newcommand{\junyao}[1]{}

\fi

\begin{document}

\maketitle

\begin{abstract}
We study the communication complexity of incentive compatible auction-protocols between a monopolist seller and a single buyer with a combinatorial valuation function over $n$ items. 
Motivated by the fact that revenue-optimal auctions are randomized~\cite{Thanassoulis04,ManelliV10,BriestCKW10,Pavlov11,HartR15} (as well as by an open problem of Babaioff, Gonczarowski, and Nisan~\cite{BabaioffGN17}), we focus on the {\em randomized} communication complexity of this problem (in contrast to most prior work on deterministic communication).

We design simple, incentive compatible, and revenue-optimal auction-protocols whose expected communication complexity is much (in fact infinitely) more efficient than their deterministic counterparts. 

We also give nearly matching lower bounds on the expected communication complexity of approximately-revenue-optimal auctions. These results follow from a simple characterization of incentive compatible auction-protocols that allows us to prove lower bounds against randomized auction-protocols.  
In particular, our lower bounds give the first approximation-resistant, exponential separation between communication complexity of {\em incentivizing} vs {\em implementing} a Bayesian incentive compatible social choice rule, settling an open question of Fadel and Segal~\cite{FadelS09}.
\end{abstract}

\setcounter{page}{0}
\thispagestyle{empty}
\newpage
\section{Introduction}

The central goal of Algorithmic Mechanism Design is to design mechanisms
that guarantee good outcomes while taking into account both (i) the
selfish agents' incentives and (ii) the ever-increasing complexity
of modern applications. A fundamental question to this field is whether
simultaneously satisfying both the incentive and simplicity constraints
is harder than satisfying each of them separately. 

In this paper we focus on one of the simplest
and most studied settings in the field: a monopolist, Bayesian, revenue-maximizing
seller auctioning $n$ items to a single risk-neutral buyer.
An active line of work over the past two decades argues that even in this strategically-simple setting, 
and even for buyers with additive or unit-demand valuations%
\footnote{To circumvent some worst-case pathological examples, 
it is common in Algorithmic Mechanism Design to restrict the buyer's value distribution to {\em independent} (vs {\em correlated}) items, and/or restrict the combinatorial nature of buyer's value for bundles to one of the following classes:
$$ \text{\em additive, unit-demand} \subset \text{\em gross-substitutes} \subset \text{\em submodular} \subset \text{\em XOS} \subset \text{\em subadditive}.
$$
(See Section~\ref{sub:valuations} for definitions.)},
optimal mechanisms are inherently complex, e.g.~they involve randomized
lotteries~\cite{Thanassoulis04,ManelliV10,BriestCKW10,Pavlov11,HartR15}
and are often computationally intractable~\cite{ChenDOPSY15,DaskalakisDT14,ChenDPSY14}.

One particularly influential measure of complexity of mechanisms is
the {\em menu-size complexity} of \cite{HartN13}: by the taxation
principle, a general incentive compatible mechanism can be canonically
represented as a {\em menu}, where each {\em line} or option in
the menu corresponds to a (possibly randomized) allocation and a payment.
The menu-size complexity of a mechanism is then the number of lines
in the corresponding menu. Perhaps the single most convincing evidence
for the complexity of optimal mechanisms is an example due to \cite{DaskalakisDT17},
where the optimal mechanism for an additive buyer with two i.i.d.~item
valuations from a seemingly benign distribution ($\Beta(1,2)$) requires
an infinite and even {\em uncountable} menu-size complexity. We
henceforth refer to this powerful example as the DDT example. 

\cite{DaskalakisDT17} and related complexity results for optimal
revenue-maximizing auctions have inspired fruitful lines of work that
circumvent these barriers, e.g.~by designing sub-optimal but simple
mechanisms that approximate the optimal revenue (see discussion in
Related work).

It is not a-priori clear, however, that the menu-size complexity by
itself is an obstacle to using optimal mechanisms. For instance, the
seller in the DDT example could in principle succinctly describe her%
\footnote{Throughout the paper, we use feminine pronouns for the seller and masculine for the buyer.}
mechanism as ``the-optimal-auction-for-$\Beta(1,2)\times\Beta(1,2)$''
and even point the buyer to an explicit description in \cite{DaskalakisDT17}.
However, \cite{BabaioffGN17} recently observed that, once the mechanism
is announced, the deterministic communication complexity to implement
it is equal (up to rounding) to the logarithm of the menu-size complexity.
In the DDT example, for the buyer to deterministically specify his
favorite line in the uncountable menu, he would need to send an infinite
stream of bits. \cite{BabaioffGN17} left open the question of randomized
communication complexity of optimal mechanisms. Indeed randomized communication
is a natural complexity measure in this case since we already consider
randomized allocations\footnote{Different applications have different simplicity desiderata. (E.g.
highly regulated FCC auctions {\em vs} very fast ad auctions with
automated bidders {\em vs} smart contracts that require costly documentation
of transaction details on a blockchain.) Ultimately, there is no universal
``right'' measures of complexity, and studying a variety gives us
a more complete understanding.}.

In this paper, inspired by \cite{BabaioffGN17}'s open question, we
formulate a notion\footnote{Technically, our definition of IC auction protocol is a special case
of {\em Bayesian incentive compatibility (BIC)-incentivizable  binary dynamic mechanism (BDM)}
\cite{FadelS09}. We discuss this connection further in Related work.} of an {\em incentive compatible (IC) auction protocol}, which is
a two-party (possibly randomized) interactive communication protocol
between a seller and a buyer with an allocation and payment associated
with every transcript of the protocol. Before presenting our results
in this model, below we briefly discuss our modeling assumptions;
a full definition appears in Section \ref{sec:Model}.

\subsection*{Brief discussion of modeling assumptions}

Per the discussion above, we assume that the protocol and auction
format are public information. The buyer privately knows his true
type (or valuations of items/bundles). 

We mostly focus on the total expected communication complexity of
the protocol. For our protocols, we bound the interim expectation, i.e.~for every
buyer's type, the communication complexity of the protocol is bounded,
in expectation\footnote{In expectation vs high probability: We remark that by Markov's inequality
in expectation {\em upper} bounds on the communication complexity
imply similar upper bounds w.h.p.; e.g.~if the expected complexity
is at most $C$, then it is at most $C/\alpha$ w.p. $\geq1-\alpha$.} over the protocol's randomness\footnote{In fact, all our protocols happen to satisfy a slightly stronger desideratum
where all the communication complexity bounds that we prove also (approximately)
hold for the communication complexity of the future of any prefix
of the protocol. I.e.~for any setting where we bound the expected
communication complexity by $C$, it is also true that, conditioning
on any history of the protocol (possibly much longer than $C$), the
remaining expected communication complexity is $O(C)$. This means
that the buyer and seller {\em always} expect -for every run of
the protocol, and at any point during the execution- that the protocol
will end soon.}. Our lower bounds hold even for ex ante
expectation, i.e.~even if we allowed that some buyers may know in advance that
they are expected to participate in a prohibitively long protocol.

The seller in our model has {\em no private information} and is {\em not strategic}.
At the end of the communication protocol she must know the allocation
and payment. 

We model the buyer's strategic aspect as a complete information single-player
extensive-form game with buyer's nodes and nodes of Chance; each leaf
is associated with an allocation and a payment. In practice, nodes
of Chance could be implemented by a trusted seller (e.g.~when the
seller is an auditable firm), a trusted intermediary, a cryptographic
protocol for coin tossing\footnote{We're mostly interested in mechanisms that are {\em exactly} revenue-optimal,
while the security of cryptographic protocols always has a negligible
but non-zero chance of being broken even by naive brute-force algorithms.
In theory, this small chance of cheating on the coin tosses would
violate the buyer's exact incentive constraints.}, or a publicly observable, renewable\footnote{By ``renewable'' we mean that at each step of the protocol the parties
have access to fresh random bits not predictable in previous iterations;
for example, they could look at the weather each day.} external source of randomness. 

As is common in the aforementioned literature on randomized mechanisms,
we assume that the buyer is risk-neutral. In particular, we require
that the protocol is interim individually rational. In direct revelation
mechanisms, it is possible to transform interim to ex-post individual
rationality by correlating the payment with the randomized allocation.
Similarly, at the cost of a bounded increase in the communication
complexity, it is possible to transform our protocols to become ex-post
approximately individually rational (see Appendix \ref{sec:Approximately-ex-post-IR}
for details).

While we make little restrictions on buyer valuations, we do generally
assume that the buyer's valuation is capped at some arbitrarily large
value $U$. The complexity of our protocols\footnote{Except the ex-post approximately individually rational protocols in
Appendix \ref{sec:Approximately-ex-post-IR}. } does not depend on $U$, e.g.~$U$ can be all the money in the universe
(typically much smaller).

\subsection*{Our IC auction protocols}

We design IC auction protocols that are simple, surprisingly efficient,
and are {\em exactly} revenue-optimal. For instance, in Theorem \ref{thm:DDT}
we give a revenue-optimal IC auction protocol for the DDT example
where the {\em buyer sends less than two bits in expectation}. (In
contrast, for a deterministic auction selling two items separately,
merely specifying the allocation requires the buyer to send two bits!)

\paragraph{Main positive result}

Our main positive result is a generic transformation of an arbitrary (revenue-optimal
or otherwise) IC and IR mechanism for additive, unit-demand, or general
combinatorial valuations to an IC auction protocol that uses $O(n\log(n)),O(n\log(n)),O(2^{n}n)$
bits in expectation respectively. We note that our protocols work
for correlated prior distributions, and even for non-monotone and
negative valuations\footnote{We assume for simplicity that all payments are non-negative.}. 
\begin{thm*}
[See Theorems~\ref{thm:additive} and~\ref{thm:truthful-B}] \hfill

For any
prior $\mathcal{\ensuremath{D}}$ of Buyer's (additive/unit-demand/combinatorial)
valuations over $n$ items bounded by maximum valuation $U$, and
any IC mechanism $\mathcal{\ensuremath{M}}$, there is an IC auction
protocol with the same expected payment and allocation, using ($O(n\log n)$/$O(n\log n)$/$O(2^{n}n)$)
bits of communication in expectation.
\end{thm*}


\paragraph{Trading off revenue for even better communication efficiency}

We obtain an exponentially more efficient protocol for
the special case of unit-demand with {\em independent items}.
Specifically, at the cost of an $\varepsilon$-fraction loss in revenue,
we obtain an IC auction protocol that uses only $\polylog(n)$ communication. 
\begin{thm*}
[See Theorem~\ref{thm:symmetric}] Let $\mathcal{D}$ be a distribution
of independent unit-demand valuations over $n$ items bounded by maximum
valuation $U$. Then, for any constant $\varepsilon>0$, there is
a $(1-\varepsilon)$-approximately revenue-optimal IC auction protocol
using $\polylog(n)$ bits of communication in expectation. 
\end{thm*}
Exhibiting the richness of our IC auction protocol model, this protocol
is substantially different from the generic transformation in our
main result, and builds on the recent {\em symmetric menu-size complexity}
of \cite{KothariMSSW19}.
\begin{rem}
For simplicity of presentation we focus on the expected communication
complexity. Here we briefly remark that our protocols also have desirable
properties in terms of round- and random-coin-complexities. For round
complexity, our protocols use $O(\log(n))$ rounds in expectation
($O(n)$ for general combinatorial valuations). For the protocols
in Theorems \ref{thm:additive} and \ref{thm:truthful-B}, it will
be easy to see how (using trivial batching) one can further compress
the number of rounds: at the cost of a constant factor increase in
the communication complexity, these protocols can be compressed to
$1+\varepsilon$ rounds in expectation. In terms of random coins,
our protocols can be implemented with $O(\log(n))$ coins in expectation
($O(n)$ for general combinatorial valuations). 
\end{rem}

\subsection*{Communication complexity lower bounds}

We show that beyond the (important) special case covered by Theorem~\ref{thm:symmetric}, the communication complexity of our protocols is almost the best possible, in the following strong sense:

\begin{thm*}[See Theorems~\ref{thm:lower_bound_unit_demand},~\ref{thm:lower_bound_gross_substitutes}, and~\ref{thm:lower_bound_XOS_ind}]
For revenue maximization with $n$ items, any incentive compatible auction protocol that achieves any constant factor approximation of the optimal revenue must use at least:
\begin{itemize}
    \item $\Omega(n)$ communication for unit-demand valuations;
    \item $2^{\Omega(n^{1/3})}$ communication for gross substitutes valuations;
    \item $2^{\Omega(n)}$ for XOS valuations.
\end{itemize}  
Furthermore, any incentive compatible auction  protocol obtaining more than $80\%$ of the optimal revenue must use at least:
\begin{itemize}
\item $2^{\Omega(n)}$ communication for XOS valuations over {\em independent items}.
\end{itemize}
\end{thm*}

To place the result for independent items in the greater context of Algorithmic Mechanism Design, 
contrast it with simple-but-approximately-optimal mechanism independent subadditive valuations:\
\cite{RubinsteinW18} showed that
a constant fraction of revenue can be guaranteed by simple mechanisms;
this constant has been improved in followup works~\cite{CaiDW16,ChawlaM16,CaiZ17}, 
but no non-trivial upper bound on the best approximation factor were
known\footnote{Note that this is a maximization problem, so {\em upper bound} on
the approximation factor refers to an impossibility result.}. Assuming that efficient randomized communication is a {\em necessary}
desideratum for ``simple mechanism'', our result for independent
items implies that the optimal approximation factor is bounded away
from 1 -- even for the special case of XOS valuations.

Note also that our upper and lower bounds for correlated valuations are nearly tight in the following ways:
\begin{itemize}
\item For unit-demand and combinatorial valuations, our upper and lower bounds nearly match (up to logarithmic factors), even though the lower bounds hold for {\em arbitrary (constant) approximation factor} vs {\em exactly} revenue-optimal in upper bounds. Furthermore the combinatorial upper bound holds for {\em arbitrary} combinatorial valuations, which are much more general than XOS valuations used in the lower bound.
\item The correlation in our unit-demand lower bound is necessary by Theorem~\ref{thm:symmetric}. 
\end{itemize}

We remark that for one interesting case an exponential gap remains:

\begin{OQ}
What is the randomized communication complexity of exactly revenue optimal IC auction protocols for unit demand valuations over independent items?
\end{OQ}

Our lower bound for unit-demand requires correlated items (and this is an inherent limitation of our technique). On the other hand, our protocol for unit-demand with independent items (Theorem~\ref{thm:symmetric}) does not guarantee exact revenue optimality.

\subsection*{Separating the complexity of implementing and incentivizing}


Our results also have implications for a question of Fadel and Segal~\cite{FadelS09}. They study, for any fixed social choice rule, the {\em communication cost of selfishness},
i.e. the difference in communication complexity between (i) implementing
it, and (ii) implementing it in a Bayesian incentive compatible protocol.
They give examples where the communication cost of selfishness is
exponential, but those examples are very brittle in the sense that
they rely on agents' utilities to have unbounded (or at least exponential)
precision. They ask whether the communication
cost of selfishness on any (possibly contrived) social choice rule
can be reduced substantially if agents' utilities have a bounded precision~\cite[Open Question 3]{FadelS09}.
Our source of hardness is inherently different from the instances
in~\cite{FadelS09}: we harness the combinatorial structure of the valuations
rather than exploiting the long representation of high-precision numbers.

In more detail, in our constructions the buyer's utility only requires constant
precision\footnote{We require constant precision marginal contribution per item. For
unit-demand, this translates to constant precision for any outcome.
For gross substitutes, etc.~this translates to $O(\log(n))$ bits
to represent outcome utilities, which is still negligible.} for any outcome (and the seller is not strategic, i.e. she has constant
utility zero). Furthermore, for our hard instances of unit-demand
 valuations, we show (Remark~\ref{remark:fadel_segal}) that the exactly
revenue-optimal IC mechanism can be implemented by a randomized (non-IC) protocol
using $O(\log(n))$ communication even in the worst
case, hence resolving~\cite{FadelS09}'s open question on the negative%
\footnote{Note that it was an open question to obtain such a separation for {\em any} social choice rule, let alone a natural and important one like revenue-maximizing auctions.}. 
We remark that by~\cite[Corollary 3]{FadelS09}, this exponential separation is tight.

\begin{corollary}[See Remark~\ref{remark:fadel_segal}]\label{cor:fadel_segal}
There exists a randomized protocol for a revenue maximization instance, in which the buyer's valuation has constant precision, such that there is an exponential separation between the communication complexity of its approximately Bayesian IC implementation and that of its non-IC implementation.
\end{corollary}

\begin{remark}[Separations for deterministic vs randomized protocols] \hfill 

Formally,~\cite{FadelS09} phrase their open question for deterministic protocols. 
To view Corollary~\ref{cor:fadel_segal} in this context, note that in our model  the seller is not strategic; hence one can consider an equivalent deterministic social choice rule in a slightly different setting where the random seed (only $O(\log(n)$ bits are necessary) to the revenue-optimal auction is replaced by a seller's type.
The requirements from the protocol in this setting is only stricter, so the communication lower bound on IC auction protocols trivially extends. On the other hand, for the non-IC auction protocol the seller can just send the buyer her type (aka the random seed).
\end{remark}

Interestingly, this separation between the communication complexity
of implementing and incentivizing optimal auctions holds in a more
general sense (albeit for expected communication in randomized protocols):
In Appendix~\ref{sec:non-IC} we show a {\em non-IC} auction protocol\footnote{The non-IC auction protocol is closer to \cite{FadelS09}'s notion
of {\em implementing} (as opposed to incentivizing) a mechanism,
or to \cite{BabaioffGN17}'s definition of randomized communication
complexity of auctions.} that for
{\em any} buyer with unit-demand (resp. combinatorial) valuations,
the exactly optimal IC mechanism can be implemented by a randomized
(non-IC) protocol using $O(\log(n))$ (resp. $O(n)$) communication. 

\subsection*{Technical highlights: infinitely more efficient auction-protocols}

Abstracting away the game theory and other detail, we explain the
simple idea which is at the core of our main positive result (Theorems \ref{thm:additive}
and \ref{thm:truthful-B}). Simplifying further, consider a randomized
auction of just a single item: our goal is to compress the infinite
deterministic communication complexity of a protocol where the buyer
tells the seller exactly with what probability he expects to receive
the item. Denote this probability of allocation by $p$. Given $p$,
one way to allocate with probability $p$ using unbiased coin tosses
is to generate a uniformly random number $\tau\in[0,1]$ (whose binary
representation is a uniformly random stream of bits after the decimal
point), and to allocate the item iff $p>\tau$\footnote{For historical context, we remark that the setup up to this point
is similar to the 1-bit public-coin protocol for single-item auctions
in \cite{BabaioffGN17}.}. 

The key insight: for any fixed $p$, we don't actually need to know
$\tau$ to infinite precision - we only need to know the prefix of
$\tau$'s binary representation until the first bit on which it differs
from $p$. Similarly, for a fixed $\tau$, we only need to know $p$
to the same precision. So here is our core protocol: draw\footnote{Here and in all our protocols, $\tau$ can be drawn on the fly so
the expected number of random bits is also bounded.} $\tau\in\left[0,1\right]$ uniformly at random, and ask the buyer
to stream the binary representation of $p$ - only with enough precision
to determine whether $p>\tau$. Each time the buyer sends a bit from
the binary representation of $p$ it differs from the corresponding
bit of $\tau$ with probability $1/2$; i.e.~the protocol terminates
with probability $1/2$ after each round. Hence we reduced the infinite
deterministic protocol to one where the buyer only sends $2$ bits
in expectation.

What happens when we bring back incentives? It's not too hard to show
that the protocol remains incentive-compatible as long as the buyer
doesn't learn anything about $\tau$ until the end of the protocol.
This is actually too good to be true, since the protocol length must
depend on $\tau$ (otherwise it would be deterministic - and hence
infinite), and the buyer must know whether the protocol is continuing
in order to participate. Fortunately we can argue that if the only
thing the buyer learns about $\tau$ is that the protocol is continuing,
this information cannot help him cheat. Intuitively, he has already
committed to the prefix of the protocol, and the extension of his
strategy for the rest of the protocol is optimal conditioned on actually
being asked to use it.

\subsection*{Technical highlights: a characterization of randomized IC auction protocols}

It is natural to try to prove communication lower bounds of IC auction protocols via a modular approach
of: (i) use Game Theory to define a restricted communication problem
that we have to solve in order to obtain near-optimal revenue; and then
(ii) use standard techniques from Communication Complexity (e.g.~a
reduction from Set Disjointness). This approach has worked successfully
in other applications of communication complexity to game theory (e.g.~\cite{PapadimitriouSS08,Dobzinski16b,ImmorlicaLMST18,GoosR18}). 
However, our non-IC auction protocol in Appendix~~\ref{sec:non-IC} formally
precludes such a modular approach because there {\em is} an efficient
communication protocol that exactly solves the game theoretic problem
we are after. (In other words, the modular approach cannot separate the communication complexity of incentivizing and implementing a social choice rule.) Instead we need to simultaneously consider the complexity
and incentives constraints, in particular we need to consider the
joint evolution of the buyer's prior and incentives in an arbitrary
randomized protocol. 

Our main novel insight is the following simple characterization of incentive compatible communication protocols:
In a general communication protocol, each buyer's node can partition the buyer's
types in an arbitrary way. But for IC protocols, the buyer's next bit is fully determined by his respective value for the expected allocations conditioned on sending ``0'' or ``1''; this means that it can only partition the buyer's types into halfspaces in valuation space (see Figure~\ref{fig:protocol_tree}). Thus IC mechanisms are much less expressive. 

The second part of the proof combines tools from Auction Theory and
Error Correction Codes to construct, for each class of valuations,
a family of priors whose (approximately) optimal mechanisms are all
different. Finally, a simple counting argument shows that the total
number of short IC protocols that satisfy our characterization is
too small to cover all the different mechanisms. 

\subsection*{Related work}

For general social choice settings, \cite{FadelS09} define {\em binary
dynamic mechanism (BDM)}, which formalizes the notion of communication
protocol between multiple agents with outcomes and payments associated
with the protocol-tree leaves. They contrast the communication complexity
of {\em incentivizable} vs {\em implementable} BDMs. Our notion
of IC auction protocols is equivalent to requiring that the protocol
is incentivizable. 

One subtle difference between our model and \cite{FadelS09} is that
the latter define BDMs as deterministic, while we focus on randomized
protocols. In our context we can encode the seller's random number
source as her type\footnote{Formally they only define finite BDMs, but they also discuss the natural
infinite variant \cite[Appendix B.1]{FadelS09}.}. In this sense, our IC auction protocol is a special case of their
{\em Bayesian incentive compatibility (BIC)-incentivizable  BDM}.
But this view misses the distinction between trusting a Bayesian prior
about other agents valuations and behavior and merely trusting the
source of randomness.

Our paper resolves an open question from~\cite{FadelS09} of separating the communication complexity of incentivizing and implementing {\em Bayesian} incentive compatible social choice rules. 
Very recently,~\cite{DR21, RSTWZ21} resolved a different open question from the same paper about separating the communication complexity of incentivizing and implementing {\em ex-post} incentive compatible social choice rules. \cite{RSTWZ21} also separate the communication complexity of ex-post vs dominant strategy incentive compatibility.

Our paper exhibits a strong separation between the communication complexity
associated with direct revelation and general mechanisms. Related
separations have been shown before by \cite{ConitzerS04b} and \cite{Dobzinski16b}
for social-welfare maximization with two or more strategic buyers.
Specifically, \cite{ConitzerS04b} show an exponential gap between
the communication complexity of direct revelation versus interactive
mechanisms. \cite{Dobzinski16b} shows that in several important settings,
the ``taxation complexity'' of deterministic mechanisms is approximately
equivalent to the communication complexity, but exhibits an exponential
gap between the two for truthful-in-expectation mechanisms. In contrast,
we consider revenue maximization with a single strategic buyer and
as few as two items. Arguably, the separation for a single strategic
buyer in our settings is more surprising since he communicates with
a seller who doesn't receive exogenous private information. More generally,
communication complexity of (approximate) social welfare maximization
in auctions with multiple strategic buyers has been extensively studied
for combinatorial auctions \cite{NisanS06,Segal07,DobzinskiNO14,Dobzinski16a,Dobzinski16b,BravermanMW16,Assadi17,BravermanMW18,EzraFNTW19,AssadiKSW20}
and related settings \cite{BlumrosenNS07,PapadimitriouSS08,BabaioffBS13,BlumrosenF13,DobzinskiD13}. 

Our paper is inspired by a discussion in \cite{BabaioffGN17} about
the communication complexity of revenue-maximizing auctions. They
prove that in general the deterministic communication complexity is
equivalent (up to rounding) to the logarithm of the menu-size complexity.
They also define a measure of randomized communication complexity
of an auction, which is most closely related to \cite{FadelS09}'s
weaker notion of implementable protocols. They give a randomized protocol
for implementing {\em any}\footnote{Note that the (revenue-)optimal auction for a single item is already
deterministic and uses only 1 bit of communication.} incentive-compatible auction for selling a single item using 1 bit
of communication and (possibly infinitely many) public random coins.

Our protocols circumvent the intractability of exactly communicating payments
(to infinite precision) by replacing them with random payments while
preserving expectation. Related ideas have been used before in algorithmic
mechanism design, e.g.~by \cite{ArcherPTT03,BabaioffKS15}.

The study of communication complexity in economics has its roots in
classic works of \cite{Barnard38} and \cite{Hayek45}. Early mathematical
formulations of the question were given by \cite{Hurwicz60,MountR74,Reichelstein84}.
Outside of auctions, communication complexity has also been considered
in AGT in the context of voting rules \cite{ConitzerS05,ProcacciaR06a,CaragiannisP11,ServiceA12a}
equilibrium computation \cite{ConitzerS04,HartM10,RoghgardenW16,BabichenkoR17,GoosR18,GanorS18,GanorKP19}
fair division \cite{Segal10,BranzeiN19,PlautR19}, interdomain routing
\cite{LSZ11}, and stable matching \cite{GonczarowskiNOR19}.

Since the seminal \cite{HartN13}, menu-size complexity has been further
studied by \cite{DaskalakisDT17,BabaioffGN17,SaxenaSW18,Gonczarowski18,KothariMSSW19}.
For a buyer with additive valuations over independent items, \cite{BabaioffGN17}
prove\footnote{Theorem 1.2 of \cite{BabaioffGN17} states a slightly weaker bound
of $n^{O(n)}$; the stronger bound is suggested in Footnote 3 of their
paper.} an $\log^{O(n)}(n)$ upper bound on the menu-size complexity of approximately-optimal
mechanisms. In this special case, this translates to an upper bound
of $O(n\log\log(n))$ on the deterministic communication complexity
- slightly more efficient than our $O(n\log(n))$ upper bound on randomized
communication complexity\footnote{The results are incomparable: \cite{BabaioffGN17} uses deterministic
communication, whereas our protocol gives {\em exact}-revenue-optimality
and allows for {\em correlated} valuations. In particular, note
that in our setting $O(n\log(n))$ is tight up to $O(\log n)$ factor,
whereas for approximate revenue with independent valuations, the true
answer (even for deterministic communication) is conjectured to be
$O(\log(n))$ \cite[Footnote 4]{BabaioffGN17}.}. Our proof is arguably much simpler\footnote{The main technical hurdle for \cite{BabaioffGN17} is a reduction
to the case where the valuations are (almost) bounded by some large
number $H=\poly(n,\varepsilon)$ with only a negligible loss in revenue;
we simply assume that the valuations are bounded by $U$, but it can
be {\em arbitrarily large}. We remark that if we assume that the
optimal mechanism obtains finite revenue (as is assumed in \cite{BabaioffGN17};
see Footnote 6 of their arXiv version), then it is easy to argue that
for any $\varepsilon>0$, capping the valuations by a sufficiently
large $U(\varepsilon)$ preserves a $(1-\varepsilon$)-fraction of
the revenue (see Appendix \ref{sec:Approximately-optimal-revenue}
for details).}. \cite{Gonczarowski18} explores the dependence on $\varepsilon$
in the menu complexity of mechanisms with additive-$\varepsilon$-suboptimality
in revenue; his main result, combined with \cite{HartN13}, implies
a $\Theta(\log(1/\varepsilon))$ bound on the deterministic communication
complexity with two items. 

For a buyer with unit-demand valuations over independent items \cite{KothariMSSW19}
define a related notion of symmetric menu-size complexity which counts
the number of lines up to symmetries, and prove an $n^{\polylog(n)}$
upper bound on the symmetric menu-size complexity. We use a slightly
stronger notion of partition-symmetric menu-size complexity (see Definition
\ref{def:PSMSC}); the bound of \cite{KothariMSSW19} also holds for
this stronger definition. We use this result for our nearly-revenue-optimal
IC auction protocol. This provides further evidence that the relatively
new notion of (partition)-symmetric menu complexity is a natural complexity
measure for auctions.

Over the past decade, computational and menu-size complexity results
of optimal auctions have motivated the design of sub-optimal but simple
mechanisms that approximate the optimal revenue~\cite{ChawlaHK07, HartlineR09, ChawlaHMS10, KleinbergW12, LiY13, BabaioffILW14,ChawlaMS15, RubinsteinW18, Rubinstein16a,ChawlaM16,HartN17,CaiZ17,HartR17,ChengGMW18,RubinsteinW18,ChawlaTT19}
or require resource augmentation~\cite{RoughgardenTY12,EdenFFTW17b,LiuP18,FeldmanFR18,BeyhaghiW19}.
Our results suggest that, in some cases, even strong menu-size complexity
lower bounds do not preclude efficient optimal mechanisms. 

\section{Model and definitions}\label{sec:Model}

Our main notion in this paper is that of IC auction protocols:
\begin{defn}
[(IC) auction protocols]\label{def:auction-protocol}

An {\em auction protocol} consists of: 
\begin{itemize}
\item A (possibly infinite) binary tree whose internal nodes are labeled
either B (for Buyer) or C (for Chance).
\item Each node of Chance has an associated probability distribution over
its children.
\item Each leaf node has an associated (non-negative) payment and (feasible)
allocation.
\item A suggested mapping from Buyer's types to Buyer's strategies, where
a {\em Buyer's strategy} corresponds to a choice of child for each
Buyer's node.
\end{itemize}
We say that an auction protocol is {\em finite} if it is guaranteed
to terminate after a finite number of rounds with probability 1 for
every Buyer's strategy. We say that an auction protocol is {\em individually
rational (IR)} if the Buyer has a strategy that guarantees expected
payment 0 and empty allocation. We say that an auction protocol is
{\em IC} (in-expectation) if it is finite and IR, and if the Buyer
weakly prefers the suggested Buyer's strategy corresponding to his
type over any other strategy in the protocol. 

The {\em expected communication complexity}, of an auction protocol
is the expected depth of the leaf reached by a worst-case Buyer's
strategy (and in expectation over nodes of Chance). Theorem \ref{thm:DDT}
refers to the {\em expected Buyer's communication}, which only counts
the number of Buyer's nodes on the path to the leaf. 
\end{defn}

Note that the buyer's strategy can be assumed wlog to be deterministic.


\subsection{Valuation classes}\label{sub:valuations}

As is standard in the Algorithmic Mechanism Design literature, we consider buyers whose value for a bundle can be restricted to one of the following classes:

\begin{definition}[Valuation classes] \hfill

A valuation function $v:2^{[n]}\to\R_{\ge0}$ may be restricted to one of the following classes:
\begin{description}
\item[Additive] If it can be written as $v(S) = \sum_{i \in S} v_i$ for some item values $v_i$'s.
\item[Unit-demand] If it can be written as $v(S) = \max_{i \in S} v_i$ for some item values $v_i$'s.
\item[Matroid-rank] If, for some matroid $M$ and item values $v_i$, it can be written as $$v(S) = \max_{\text{$T$ is independent in $M$}} \sum_{i \in S \cap T} v_i.$$
\item[XOS]\footnote{XOS valuations are sometimes also called {\em fractionally subadditive}.}
If item value-vectors $\mathbf{v}_i$ of dimension $d$, it can be written as $$v(S) = \max_{j \in \{1,\dots,d\}} \sum_{i \in S} v_{i,j}.$$
\end{description} 
\end{definition}

The aforementioned classes are related to other well-studied classes like gross-substitutes, submodular, and subadditive in the following hierarchy:

$$\text{additive, unit-demand} \subset \text{matroid rank}\subset\text{gross substitutes}\subset\text{submodular}\subset\text{XOS} \subset \text{subadditive}.$$

The formal definition of gross substitutes, submodular, and subadditive is not important for our purposes; they are economically significant because they capture different natural notions of substitutability between items (see e.g.~\cite{LLN06}).

In general, we are interested in any prior distribution over valuations of any of above-mentioned types.
In particular, we also consider the notion of combinatorial valuations over independent items, which has been recently used by e.g.~\cite{RubinsteinW18,CaiDW16,ChawlaM16,CaiZ17}.
\begin{definition}[independent items~\cite{Sch03}]
A prior distribution $\D$ of valuations has a latent structure of independent items if there is a latent product distribution $\D_1\times\D_2\dots\times\D_n$ with arbitrary support such that, a sample valuation $v$ from $\D$ can be generated by first sampling $a_i$ from $\D_i$ for all $i\in[n]$, and then for every $S\in[n]$, the value of $v(S)$ is uniquely determined by $\{a_i\mid i\in S\}$.
\end{definition}

\subsection{menu-size complexity}
\begin{defn}
[Menu-size complexity]

By the taxation principle, any mechanism can be canonically described
by the expected allocation and payment for each type. This description
induces a {\em menu}, or collection of {\em menu lines}, where each
menu line is the expected allocation and payment for some type. The
{\em menu-size complexity} of a mechanism is the number of distinct
menu lines.
\end{defn}

\section{IC auction protocols for an additive buyer}
\begin{thm}
\label{thm:additive}For any prior $\mathcal{\ensuremath{D}}$ of
Buyer's additive valuations over $n$ items bounded by maximum valuation
$U$, and any truthful mechanism $\mathcal{\ensuremath{M}}$, there
is an IC auction protocol with the same expected payment and allocation,
using $O(n\log n)$ bits of communication.
\end{thm}

\begin{proof}
First, we convert $\mathcal{\ensuremath{M}}$ to a strategically-equivalent
mechanism $\mathcal{\ensuremath{M}}^{'}$ where the payment is always
either zero or $U$. Note that by IR, the expected payment $P$ in
$\mathcal{\ensuremath{M}}$ for every type is always at most $U$;
therefore for each type we can implement expected payment $P$ by
charging a payment of $U$ with probability $P/U$ (and zero otherwise).
We henceforth identify each type of Buyer with the corresponding vector
in $[0,1]^{n+1}$, which describes the probability that $\mathcal{M}'$
allocates each item to the Buyer, and the probability ($n+1$-th coordinate)
that the Buyer pays $U$. We can further identify the mechanism $\mathcal{M}'$
with the set of allowed types/vectors in $[0,1]^{n+1}$.

\paragraph{Buyer's nodes and suggested strategy}

Each Buyer's node\footnote{Here we slightly abuse notation: we defined the auction protocols
for binary trees, so this technically corresponds to a sub-tree of
depth $n+1$ with all Buyer's nodes.} corresponds to a choice of $n+1$ bits. Given the Buyer's type and
mechanism $\mathcal{M}^{'}$, let $p_{1},\dots,p_{n}$ denote the
probability that Buyer is allocated items $1,\dots,n$, respectively,
and let $p_{n+1}=P/U$ denote the probability that the Buyer pays
$U$. The Buyer's suggested strategy is to send, for each round $r$
and $i\in[n+1]$, the $r$-th bit in the binary representation\footnote{If $p_{i}$ has two binary representations, using either one throughout
the protocol will work.} of $p_{i}$. 

\paragraph{Correcting infeasible bits}

We enforce that at any point in the protocol, the Buyer's messages
are consistent with some type, i.e. with the prefix of probabilities
corresponding to some feasible menu line in $\mathcal{\ensuremath{M}}^{'}$.
If only possible value for the Buyer's next bit would possibly be
consistent with the protocol's history, the protocol continues assuming
that the Buyer indeed sent this bit (formally we remove the Buyer's
node from the protocol since it is redundant).

\paragraph{Nodes of Chance}

The distribution over nodes of Chance is determined by an implicit
parameter $\tau$ drawn uniformly at random from $[0,1]$. Before
each node of Chance, we will already know that $\tau$ belongs to
a particular measurable subset $S\subseteq[0,1]$. For a partition
$S_{L}\cup S_{R}=S$ (to be specified below), each child of this node
of Chance will correspond to $\tau$ falling in each of $S_{L}$ or
$S_{R}$, which induces the probability distribution on the children.
While $\tau$ plays a crucial role in defining and analyzing the protocol,
we stress that it is only implicit: in the actual protocol it is drawn
on the fly, with increasing precision at each node of Chance along
the path of the protocol.

To define the $r$-th node of Chance along a given path, consider,
for each $i\in[n+1]$, the concatenation of the $i$-th bits across
the Buyer's $r$ messages, and compare it to the first $r$ bits in
the binary representation of $\tau$. If for every $i$, at least
one of the bits is different, the protocol is terminated at a leaf
as follows (see Payment and Allocation). Otherwise, the protocol continues
in a Buyer's node. Note that for each node of Chance, only one of
its children is an internal (Buyer's) node.

\paragraph{Payment and Allocation}

At the end of the protocol, for each $i\in\left[n+1\right]$, let
$\widehat{p_{i}^{r}}\in[0,1)$ denote the number whose binary representation
is the concatenation of the $i$-th bit in each of the $r$ rounds
of the protocol (after correcting infeasible bits). For $i\in\left[n\right]$,
the $i$-th item is allocated iff $\widehat{p_{i}^{r}}>\tau$; the
Buyer pays $U$ iff $\widehat{p_{n+1}^{r}}>\tau$, and otherwise he
pays zero. 

\begin{figure}
\caption{Example protocol}

\begin{centering}
\includegraphics[scale=0.55]{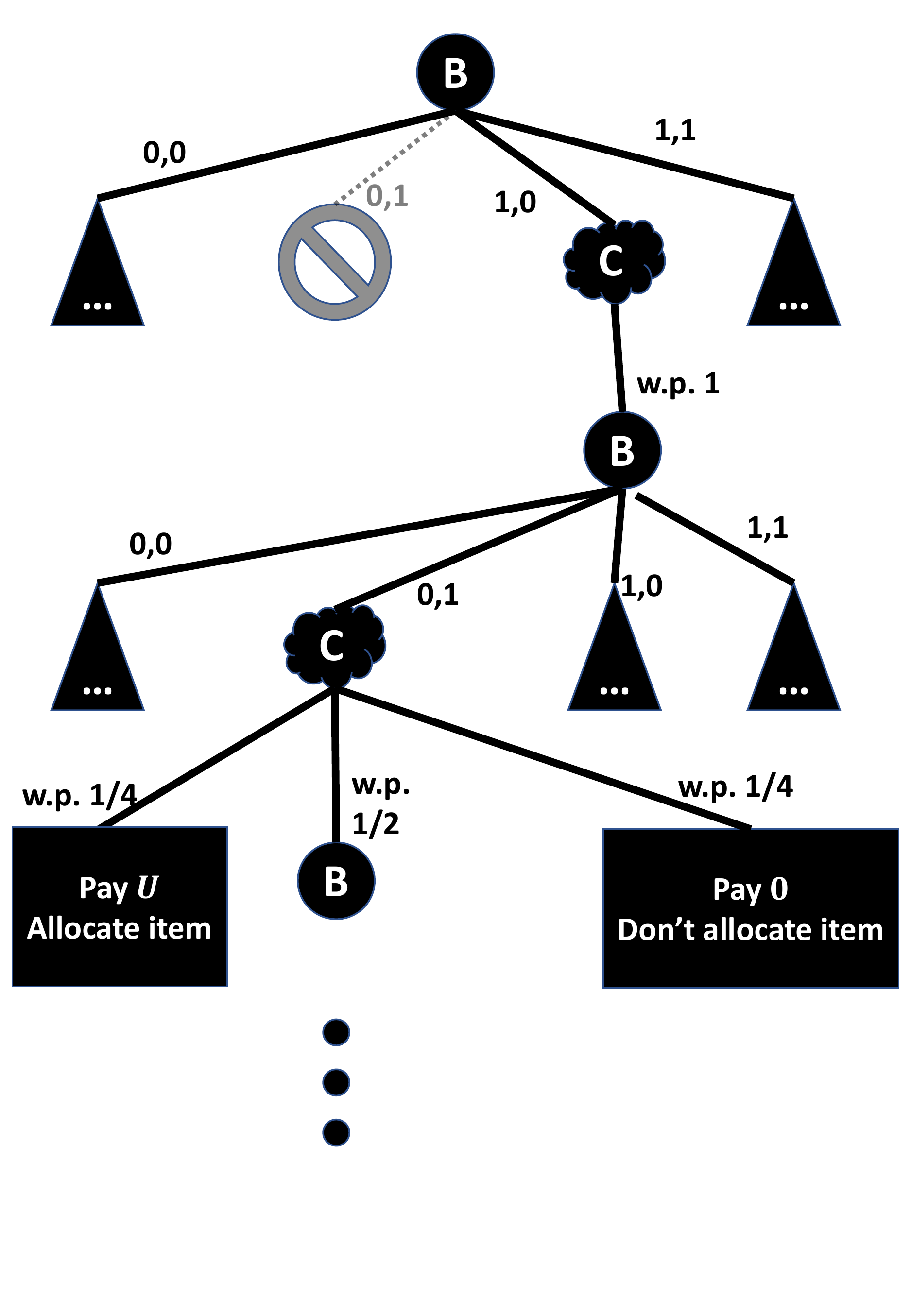}
\par\end{centering}
\vspace{-1.7cm}This figure depicts the first two iterations in an
example protocol with one item, where the Buyer's favorite menu line
has payment probability $2/3$ ($.\overline{10}$ in binary) and item
allocated with probability $1/3$ ($.\overline{01}$ in binary). Nodes
marked with B (resp. C) correspond to Buyer (resp. Chance). Triangles
correspond to sub-trees never visited for this particular Buyer's
valuation. In the first iteration, the Buyer sends $1,0$, corresponding
to the first bit in the probability of payment,allocation. Notice
that $0,1$ is an infeasible prefix for the Buyer since it would violate
IC constraints (lower probability of payment and higher probability
of allocation). At the first node of Chance, $\tau$ cannot disagree
with both bits, hence the protocol proceeds to the next Buyer's node
with probability 1. In the next iteration the Buyer sends the second
bit from each probability. Finally, in the second node of Chance:
\begin{itemize}
\item The Buyer pays $U$ and receives the item w.p. $1/4$ ($\tau<1/4<1/3,2/3$).
\item The Buyer pays nothing and receives nothing w.p. $1/4$ ($\tau>3/4>1/3.2/3$).
\item W.p. $1/2$ the protocol continues.
\end{itemize}
\end{figure}

\paragraph{IC}

The key observation for incentive compatibility is that a Buyer's
strategy is completely determined by the infinite stream of messages
that it would send in the (zero-probability) event that the protocol
never terminates. To see this, recall that each node of Chance has
only one internal node child. Hence for any fixed Buyer's strategy
there is a unique infinite path in the tree, and every finite run
of the protocol corresponds to a prefix of this path, up to some node
of Chance that deviates from the path to a leaf. 

Let $\widehat{p_{i}}$ denote the number whose binary representation
is the infinite sequence of Buyer's $i$-th bits in the (zero-probability)
event that the protocol never terminates. Recall from the previous
paragraph that a Buyer's strategy is completely determined by the
vector of $\widehat{p_{i}}$'s. Note further that the $i$-th item
is allocated at the end of the protocol iff $\widehat{p_{i}}>\tau$;
similarly, the Buyer pays $U$ iff $\widehat{p_{n+1}}>\tau$. Therefore,
since $\tau$ is drawn uniformly from $\left[0,1\right]$, the probability
that the Buyer is allocated item $i$ (resp. pays $U$) is exactly
$\widehat{p_{i}}$. Hence, by IC of $\mathcal{M}^{'}$, the suggested
strategy $\widehat{p}=p$ is optimal for the Buyer.

\paragraph{Communication complexity}

At each round of communication, the Buyer sends $n+1$ bits. Also,
at each round $r$ of communication, there is probability exactly
$1/2$ that the $i$-th bit in the Buyer's message (for each $i\in\left[n+1\right]$)
disagrees with the $r$-th bit of $\tau$. (This probability is independent
across rounds, but correlated for different $i$'s.) After $2\log(n)$
rounds, each $i$ has probability $1/n^{2}$ of agreeing with all
of $\tau$'s bits. We can take a union bound over all $i$'s to obtain
that except with probability $1/n$, the protocol has already terminated.
In the unlikely event that the protocol continues, we can re-apply
the same analysis from scratch. 

Let $\ubr$ denote an upper bound on the expected number of rounds
in the protocol, corresponding to the worst case where the above union
bound is tight. Then we have that 
\begin{equation}
\ubr\leq2\log(n)+\ubr/n.\label{eq:recursive}
\end{equation}
Solving the recurrence relation for $\ubr$, we have that $\ubr=O(\log(n))$.
Since the Buyer sends $n+1$ bits in each round, the total communication
complexity is $O(n\log(n))$.
\end{proof}

\subsection{Special case: a protocol for the \cite{DaskalakisDT17} example}

We can also prove a concrete (non-asymptotic) bound on the expected
number of bits that the buyer sends in the DDT example. Beyond the
historical importance of this specific example, our result demonstrates
that our protocols are communication-efficient not only in the asymptotic
sense, especially if we take advantage of the particular features
of a specific distribution. We in particular highlight the fact that
the Buyer in this protocol sends {\em strictly less} bits\footnote{Here we only count communication from the Buyer and not the random
coin tosses. In many scenarios random bits are cheap but informative
communication is costly.} than he would with a simple deterministic auction selling each item
separately.
\begin{thm}
\label{thm:DDT} Consider the case of $n=2$ items and the Buyer drawing
his valuations i.i.d. from $\Beta(1,2)$ (i.e.~the distribution on
$[0,1]$ with density function $f(x)=2(1-x)$). Then there is an IC
auction protocol obtaining the maximum possible revenue where the
Buyer sends less than two bits in expectation.
\end{thm}

The proof is deferred to Appendix \ref{app:DDT}.

\section{An extension for general valuations}

The following theorem is an analogue of Theorem \ref{thm:additive}
for general combinatorial valuations (not necessarily subadditive
or monotone). The communication complexity upper bound is parameterized
by $B$, the number of bundles ever assigned by the direct revelation
mechanism. For example, for unit demand valuations, $B\leq n+1$;
for general valuations, $B\leq2^{n}$.
\begin{thm}
\label{thm:truthful-B}Let $\mathcal{\ensuremath{D}}$ be any prior
over Buyer's combinatorial valuations over $n$ items bounded by maximum
valuation $U$, and any truthful mechanism $\mathcal{\ensuremath{M}}$.
Suppose that for any type and realization of randomness, $\mathcal{M}$
only ever allocates one of $B$ bundles. Then there is an IC auction
protocol with the same expected payment and allocation using $O(B\log(B))$
bits of communication.
\end{thm}

\begin{proof}
[Proof sketch] For any type, consider a partition of $[0,1]$ into
$B$ intervals, where the $b$-th interval is of length identical
to the probability that $\mathcal{\ensuremath{M}}$ allocates Bundle
$b$ to the Buyer. The rest of the proof proceeds analogously to the
proof of Theorem \ref{thm:additive}. First, we transform $\mathcal{\ensuremath{M}}$
into a mechanism $\mathcal{\ensuremath{M}}^{'}$ with payment $0$
or $U$. We henceforth identify between a type and the $B-1$ probabilities
that define the partition, and the probability that the Buyer pays
$U$. The nodes of Chance are parameterized by a threshold $\tau$
drawn uniformly at random from $[0,1]$. At each round of communication
the Buyer (allegedly) sends the next bit in the binary representation
of each of the $B$ probabilities that define his type. The protocol
terminates when it has received enough information to determine in
which of the $B$ intervals $\tau$ lies and whether $\tau$ is smaller
than the probability of payment. The allocation is the bundle corresponding
to this interval, and the payment is $U$ if $\tau$ is smaller than
the probability of payment (and zero otherwise).
\end{proof}

\section{Unit-demand, independent items: trading off revenue and communication}
\begin{thm}
\label{thm:symmetric} Let $\mathcal{D}$ be a distribution of independent
unit-demand valuations over $n$ items bounded by maximum valuation
$U$. Then, for any constant $\varepsilon>0$, there is a $(1-\varepsilon)$-approximately
revenue-optimal IC auction protocol using $\polylog(n)$ bits of communication. 
\end{thm}

Our proof uses a result of \cite{KothariMSSW19} for {\em Partition-symmetric
menus} which we introduce in Section \ref{subsec:Partition-symmetric-menu-size-co}.
The proof of Theorem \ref{thm:symmetric} is given in Section \ref{subsec:Proof-of-Theorem}.

\subsection{Partition-symmetric menu-size complexity\label{subsec:Partition-symmetric-menu-size-co}}

\subsubsection*{Symmetries}

The following is a slight strengthening of the symmetric menu-size
complexity measure recently introduced by \cite{KothariMSSW19}.
\begin{defn}
[Partition-symmetric menu-size complexity]\label{def:PSMSC}

A {\em partition-symmetric menu line} consists of a payment, (randomized)
allocation, and a partitioning of items into subsets $S_{1},\dots,S_{\sigma}$.
We say that a direct revelation mechanism $\mathcal{M}$ supports
this partition-symmetric menu line if its menu contains a line with
the same payment for any permutation of the allocation that respects
the partition (i.e.~permutation $\pi$ such that $\pi(S_{i})=S_{i}$
for all $i$). The {\em partition-symmetric menu-size complexity}
of $\mathcal{M}$ is the smallest $c$ such that $\mathcal{M}$ can
be written as the union of $c$ partition-symmetric menu lines.
\end{defn}

The following theorem follows from \cite{KothariMSSW19}; the statement
here is slightly stronger than the formulation of Theorem IV.5 in
their paper in the sense that (i) we consider the specific symmetry
group induced by a partition of the items; and (ii) we require that
the allocation probabilities are rounded to a discrete set $L_{\delta}$.
Both desiderata follow from their proof \cite{Weinberg20}.
\begin{thm}
[\cite{KothariMSSW19}]\label{thm:KMMSW}

Let $\mathcal{D}$ be a distribution of independent unit-demand valuations
over $n$ items. Then, for any constant $\varepsilon>0$, there exists
a unit-demand mechanism\footnote{We say that a mechanism is {\em unit-demand} if it never allocates
more than one item to the Buyer. (This is wlog for direct revelation
mechanisms with unit-demand buyers. But in general, for mechanisms
where the Seller does not fully learn the Buyer's valuation, it is
not obvious how to convert a mechanism where she allocates a bundle
of items to a unit-demand mechanism without increasing the partition-symmetric
menu-size complexity.)} with partition-symmetric menu-size complexity at most $n^{\polylog(n)}$
which recovers at least $(1-\varepsilon)$-fraction of the optimal
revenue. Furthermore, for some constant $\delta>0$ that depends on
$\varepsilon$, the probabilities that the mechanism allocates each
item always belong to the discrete set $L_{\delta}:=\left\{ 1,1-\delta,(1-\delta)^{2},\dots,(1-\delta)^{\frac{3}{\delta}\ln n}\right\} \cup\{0\}$;
in particular there are only $O\left(\log n\right)$ possible probabilities. 
\end{thm}

\subsection{Proof of Theorem \ref{thm:symmetric}\label{subsec:Proof-of-Theorem}}
\begin{proof}
We begin with the partition-symmetric mechanism of \cite{KothariMSSW19}
(see Theorem \ref{thm:KMMSW}). Denote its partition-symmetric menu-size
complexity by $C$. In the first stage of the protocol, the Buyer
chooses a partition-symmetric menu line among $C$ options, and then
a subset $S_{i}$ is drawn by Chance from the $\sigma\leq n$ subsets
in the partition. (Each subset $S_{i}$ is drawn with probability
equal to the sum of probabilities of items in that subset.) This first
stage uses $O(\log n+\log C)$ communication. We henceforth focus
on implementing the mechanism restricted to $S_{i}$. I.e.~a mechanism
whose menu has a fixed payment $P$ and the set of feasible allocations
is symmetric with respect to any permutation of $S_{i}$. 

Since the set of feasible allocations is symmetric, it suffices to
consider the histogram of allocation probabilities. The Buyer may
assign each probability from the histogram to any item in $S_{i}$.
Recall also that by Theorem \ref{thm:KMMSW}, all the probabilities
in the histogram belong wlog to a discrete set $L_{\delta}$ of $O(\log(n))$
feasible probabilities. In particular, the histogram can be described
by $O(\log^{2}(n))$ bits (since the count for each probability is
an integer between $0$ and $|S_{i}|\leq n$).

The second stage of the protocol proceeds by recursively considering
smaller subsets of $S_{i}$. The nodes of Chance are parameterized
by a number $\tau$ draws uniformly at random from $[0,1]$. At the
first iteration, the Buyer's suggested strategy is to send the histogram
of probabilities for the lexicographically first half of items in
$S_{i}$. (This is equivalent to sending the histogram for the second
half of the items since the total histogram is known.) If the sum
of probabilities in the first half is greater than $\tau$, the protocol
recurses on the first half; otherwise it recurses on the second half.
After $O(\log|S_{i}|)$ iterations, only one item is left. The Buyer
is allocated that item and pays $P$.

\paragraph*{IC}

We prove that the second stage of the protocol is IC and also has
the same expected allocation and payment as in the original mechanism;
IC of the first stage then follows from IC of the original mechanism.
To show second-stage IC, let $\widehat{p_{j}}$ denote the probability
that the Buyer assigns item $j$ in the last iteration when it is
not the only remaining item. Observe that any Buyer's strategy for
the second stage is fully determined by the vector of $\widehat{p_{j}}$'s.
By reverse induction over the iterations of the protocol, observe
that the histogram of all $\widehat{p_{j}}$'s is exactly equal to
the histogram of feasible probabilities. Finally note that at the
end of the protocol, the Buyer is allocated item $j$ with probability
$\widehat{p_{j}}$. Therefore, by IC of the original protocol, the
Buyer's suggested strategy is optimal.

\paragraph*{Communication complexity}

The first stage of the protocol requires $O(\log n+\log C)$ communication.
Each iteration of the second stage requires $O(\log^{2}(n))$ bits
to describe the histogram, and there are at most $O(\log(n))$ iterations.
Hence the total communication complexity is $O(\log^{3}n+\log C)=\polylog(n)$.
\end{proof}

\section{Communication lower bound for unit-demand valuations}
We consider revenue maximization with unit-demand valuations as an example to demonstrate our proof technique. Our framework for constructing hard instances will rely on the design properties of a set system and a vector family, which are presented in the two lemmata in the following subsection.
\subsection{Combinatorial designs}
\begin{lemma}\label{lem:weak_design}
For any constant $\varepsilon,\delta>0$, there exists a family of size-$\varepsilon n$ subsets $\X_{n,\varepsilon,\delta}\subset \{0,1\}^n$ such that $|\X_{n,\varepsilon,\delta}|=2^{\Omega(n)}$, and the intersection between any two distinct subsets $x_1,x_2\in \X_{n,\varepsilon,\delta}$ has size at most $(1+\delta)\varepsilon^2 n$.
\end{lemma}
\begin{proof}
By Chernoff bound, the size of intersection between two random size-$\varepsilon n$ subsets $S_1,S_2$ is concentrated:
\[
\P[||S_1\cap S_2|-\varepsilon^2 n|\ge \delta \varepsilon^2 n]\le e^{-\delta^2\varepsilon^2 n/3}.
\]
We draw $2^{\theta n}$ random subsets of size $\varepsilon n$, for $\theta=\delta^2\varepsilon^2/6$. Then, by a union bound over all pairs of subsets, the size of intersection between every two random subsets is concentrated with high probability.
\end{proof}
\begin{lemma}\label{lem:low_discrepancy}
For any constant $\varepsilon>0$ and large integer constant $\ell$, let $\cR_{\ell,\varepsilon}$ be the discrete distribution supported on $\{\varepsilon^{\ell-1},\varepsilon^{\ell-2},\dots,1\}$ such that $p^{(i)} \propto \varepsilon^{\ell-i}$, where we denote $p^{(i)}:=\P[\varepsilon^{i-1}]$ (this is approximately the ``equal-revenue distribution''). Then, for any constant $\eta>0$, there exists a family of vectors $\C_{N,\ell,\varepsilon,\eta}\subset \{\varepsilon^{\ell-1},\varepsilon^{\ell-2},\dots,1\}^N$ such that
\begin{itemize}
    \item $|\C_{N,\ell,\varepsilon,\eta}|=2^{\Omega(N)}$,
    \item and moreover, for any $m=\omega(1)$ distinct vectors in $\C_{N,\ell,\varepsilon,\eta}$, for all but $\eta$ fraction of $j\in [N]$, for any $i\in[\ell]$, there are $(1\pm \eta)p^{(i)}$ fraction of these $m$ vectors whose $j$-th coordinates are $\varepsilon^{i-1}$.
\end{itemize}
\end{lemma}
\begin{proof}
We construct $\C_{N,\ell,\varepsilon,\eta}$ simply by independently sampling $2^{\delta N}$ vectors from product distribution $\cR_{\ell,\varepsilon}^N$ for arbitrarily small constant $\delta>0$, and we show that the desired properties hold with high probability. First, the probability that two random vectors have the same value at $j$-th coordinate is $p:=\sum_{i\in[\ell]} p^{(i)}\cdot p^{(i)}$ for any $j$, and therefore, the probability that the two random vectors are exactly the same is $p^N$. For $\delta<\log(1/p)/2$, by a union bound over all the pairs of random vectors, every vector is distinct with high probability. Second, for any $m$ random vectors, for any $i\in[\ell],j\in[N]$, let $m_{i,j}$ be the number of vectors whose $j$-th coordinates are $\varepsilon^{i-1}$ among the $m$ random vectors, then by Chernoff bound,
\[
    \P[|m_{i,j}-p^{(i)} m|\ge \eta \cdot p^{(i)} m]\le e^{-\eta^2 \cdot p^{(i)} m/3}.
\]
By a union bound, the probability that there exists $i\in[\ell]$ such that $m_{i,j}$ is not within $(1\pm \eta)p^{(i)}m$ is at most $\ell\cdot e^{-\eta^2 \cdot p^{(i)} m/3}$. It follows that for any fixed $\eta$ fraction of $j\in[N]$, the probability that there exists $i\in[\ell]$ such that $m_{i,j}$ is not within $(1\pm \eta)p^{(i)}$ for all $j$ among the $\eta$ fraction is at most $(\ell\cdot e^{-\eta^2 \cdot p^{(1)} m/3})^{\eta N}$ (notice that $p^{(1)}$ is the smallest among all $p^{(i)}$'s). By another union bound over all possible $\eta$ fraction of $j\in[N]$, the probability that the second property in the statement is violated for $m$ random vectors is at most
\begin{align*}
    \binom{N}{\eta N} \cdot (\ell\cdot e^{-\eta^2 \cdot p^{(1)} m/3})^{\eta N}&\le (e/\eta)^{\eta N}\cdot (\ell\cdot e^{-\eta^2 \cdot p^{(1)} m/3})^{\eta N} \\
    &=((e/\eta)\cdot \ell\cdot e^{-\eta^2 \cdot p^{(1)} m/3})^{\eta N},
\end{align*}
which is $e^{-\theta m N}$ for some constant $\theta$ that does not depend on $\delta$. Since there are $\binom{2^{\delta N}}{m}\le (e\cdot 2^{\delta N}/m)^m\le e^{\delta m N}$ distinct subsets of $m$ random vectors of $\C_{N,\ell,\varepsilon,\eta}$, by union bound, for $\delta<\theta$, for any fixed $m$, the second property in the statement is violated with probability at most $e^{-(\theta-\delta) m N}$. Finally, the proof finishes by taking a union bound over all $m=\omega(1)$, namely, $\sum_{m=\omega(1)} e^{-(\theta-\delta) m N}=o(1)$.
\end{proof}

\subsection{The main lower bound result}
Now we prove the following lower bound result for communication complexity of approximate revenue maximization with unit-demand valuations. Specifically, we construct a family of priors and show that most priors in the family are hard for all low-communication (almost) truthful-in-expectation randomized protocols to approximately maximize revenue.

\begin{theorem}\label{thm:lower_bound_unit_demand}
For every constant $\tau>0$, any $\tau$-approximate (almost) truthful-in-expectation protocol for revenue maximization, where the seller has $n$ items, and the buyers have unit-demand valuations, requires $\Omega(n)$ bits of communication in expectation.
\end{theorem}

\begin{proof}
We first construct a family of prior distributions of the buyers' valuations and then argue that in order to achieve any constant approximation, a protocol tree (which we will elaborate shortly) can not be shared by many prior distributions, which implies the communication complexity lower bound by a counting argument.

\paragraph{Construction} For arbitrarily tiny constants $\varepsilon_1,\varepsilon_2,\delta_1,\eta>0$ and large integer constant $\ell$ such that $\eta,\varepsilon_1(1+\delta_1)\ll \varepsilon_2^{\ell}$, we take the set family $\X_{n,\varepsilon_1,\delta_1}$ from Lemma~\ref{lem:weak_design} and let $N:=|\X_{n,\varepsilon_1,\delta_1}|=2^{\Omega(n)}$, and then, we take the vector family $\C_{N,\ell,\varepsilon_2,\eta}$ from Lemma~\ref{lem:low_discrepancy} with $|\C_{N,\ell,\varepsilon_2,\eta}|=2^{\Omega(N)}=2^{2^{\Omega(n)}}$. We let each $x\in\X_{n,\varepsilon_1,\delta_1}$ represent a subset of items. Notice that we can fix a one-to-one mapping between the coordinates of a vector in $\C_{N,\ell,\varepsilon_2,\eta}$ and all the sets in $\X_{n,\varepsilon_1,\delta_1}$, and therefore, for any vector $c\in\C_{N,\ell,\varepsilon_2,\eta},\,x\in\X_{n,\varepsilon_1,\delta_1}$, we can denote $c(x)$ as $c$'s value at the coordinate that corresponds to $x$.

For each vector $c\in\C_{N,\ell,\varepsilon_2,\eta}$, we construct a prior distribution $\D_c$ of the buyers' valuations as follows --- First, for each $x\in\X_{n,\varepsilon_1,\delta_1}$, we define a unit-demand valuation $v^x_c:2^{[n]}\to \R_{\ge0}$ as follows: 
\[
    v_c^x(S) := 
        \begin{cases}
            0   & x \cap S=\emptyset\\
            c(x) & \text{otherwise}.
        \end{cases}
\]
Then, we let $\D_c$ be the uniform distribution over $v_c^x$'s for all $x\in\X_{n,\varepsilon_1,\delta_1}$. Finally, the family of prior distributions is $\F=\{\D_c\mid c\in\C_{N,\ell,\varepsilon_2,\eta}\}$.

\paragraph{Interpretation}
The following interpretations might be helpful for reading the proof.
Each $x\in\X_{n,\varepsilon_1,\delta_1}$ corresponds to a set of items which are (equally) valuable to the buyer with valuation $v_c^x$. Each vector $c\in\C_{N,\ell,\varepsilon_2,\eta}$ specifies for each $x\in\X_{n,\varepsilon_1,\delta_1}$ how valuable such an item is to the buyer with valuation $v_c^x$. By the design property of $\X_{n,\varepsilon_1,\delta_1}$, every $v_c^{x_1},v_c^{x_2}$ with distinct $x_1,x_2$ are interested in mostly different items. By the design property of $\C_{N,\ell,\varepsilon_2,\eta}$, for a large number of valuations $v_c^{x}$'s with distinct $c$'s but the same $x$, the values of an item in $x$ to these valuations are distributed roughly according to the ``equal revenue distribution'' $\cR_{\ell,\varepsilon_2}$ defined in Lemma~\ref{lem:low_discrepancy}.

\subsubsection*{An optimal truthful-in-expectation protocol for the hard instances}
The first step for proving the lower bound is to show that there is a truthful-in-expectation protocol that extracts the full welfare using $O(n)$ bits of communication for the family of Bayesian instances constructed above. The protocol is as follows: the buyer sends the set $x$ that corresponds to his valuation $v_c^x$ to the seller, which takes $n$ bits, and then, if $x\in\X_{n,\varepsilon_1,\delta_1}$ (otherwise the seller stops), the seller samples an item $i$ from set $x$ uniformly at random and gives the item $i$ to the buyer and charges him $c(x)$, where $c$ corresponds to the prior $\D_c$. This protocol is obviously individual rational and revenue maximizing if the buyer tells the truth. To show truthfulness in expectation, suppose the buyer's true set of interest is $x$; without loss of generality, we can assume that the buyer sends some $x'\in\X_{n,\varepsilon_1,\delta_1}$, because otherwise, the seller stops, and the buyer gets net utility 0, which is not better than telling the true $x$. Moreover, if the buyer sends $x'\neq x$, by the design property of $\X_{n,\varepsilon_1,\delta_1}$, he receives an item in $x$ with probability at most $\varepsilon_1(1+\delta_1)$. Hence in expectation, the net utility is at most $\varepsilon_1(1+\delta_1) c(x)-c(x')\le \varepsilon_1(1+\delta_1)-\varepsilon_2^{\ell-1}< 0$, where the first inequality is due to $c(x)\le 1$ and $c(x')\ge \varepsilon_2^{\ell-1}$, and the second is due to our choice of parameters. Thus, sending $x$ instead of $x'$ is strictly better in expectation.

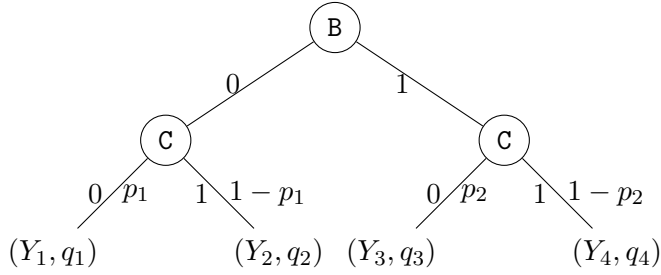
\begin{figure}
    \centering
    \begin{tikzpicture}[tree node/.style={circle,draw},
        level 2/.style={sibling distance=3cm}]]
        \node[tree node]{\texttt{B}}
            child { 
                node[tree node] {\texttt{C}} 
                child {
                    node {$(Y_1,q_1)$}
                    edge from parent node[left] {0} node[right] {$p_1$}
                }
                child {
                    node {$(Y_2,q_2)$}
                    edge from parent node[left] {1} node[right] {$1-p_1$}
                }
                edge from parent node[left] {0}
            }
            child [missing]
            child [missing]
            child { 
                node[tree node] {\texttt{C}} 
                child {
                    node {$(Y_3,q_3)$}
                    edge from parent node[left] {0} node[right] {$p_2$}
                }
                child {
                    node {$(Y_4,q_4)$}
                    edge from parent node[left] {1} node[right] {$1-p_2$}
                }
                edge from parent node[left] {1}
            };
    \end{tikzpicture}
    \caption{A depth-2 protocol tree.}
    \label{fig:protocol_tree}
\end{figure}

\subsubsection*{Representing a protocol as a protocol tree per prior distribution}
Observe that once the prior distribution is fixed, a protocol can be viewed as a protocol tree. See Figure~\ref{fig:protocol_tree} for example. Without loss of generality, the protocol tree starts with the root \texttt{B} representing the buyer's round and then alternates between the buyer \texttt{B} and the seller \texttt{C} (Chance). At each round, represented by a node, the buyer or the seller can choose to send a bit $0$, represented by left edge, or bit $1$, represented by right edge, to the other. At a leaf, both players agree on a set of items $Y$ allocated to the buyer and a payment $q$ to the seller. The protocol is possibly randomized, and hence, at a seller's round, the seller\footnote{We assume that the seller is not strategic in the private-coin model. In the public-coin model, the seller can not be strategic, because his responses can be inferred from the public randomness and the pre-specified protocol tree, and thus, he can keep silent unless he observes that the buyer is cheating.
} can send bit $0$ with probability $p$ and send bit $1$ with probability $1-p$, which are represented by the weights on the edges. At a buyer's round, the buyer's strategy depends on his valuation, but we can assume without loss of generality that the buyer always deterministically chooses a bit to send, because the buyer is strategic and hence sending the bit that has better net utility in expectation (sending the bit that maximizes the seller's revenue if both choices are (almost) equal, and sending bit $0$ if it is still a tie) is a (almost) dominant strategy for the buyer that maximizes the seller's revenue among all (almost) dominant strategies. Therefore, the buyer's prescribed (almost) dominant strategy can be deterministically decided by the protocol tree and his valuation.

To make the proof easier, we show that we can without loss of generality assume some nice properties for the protocol trees, and we will only consider such protocol trees afterwards.
\begin{claim}\label{claim:bounded_protocol_trees}
Any (almost) truthful-in-expectation protocol with $O(k)$ communication in expectation for our hard instance can be changed (with arbitrarily small loss of the approximation factor) such that
\begin{itemize}
    \item the protocol tree has $O(k)$ depth,
    \item and moreover, the payment at any leaf of the protocol tree is $2^{O(k)}$.
\end{itemize}
\end{claim}
Suppose a protocol uses $\alpha k$ bits of communication in expectation where $\alpha$ is a positive constant. For an arbitrarily large constant $\beta$, by Markov's inequality, the protocol takes $\ge\beta k$ communication with probability at most $\gamma:=\alpha/\beta$.
Observe that if we trim all the nodes at level $\ge \beta k$ of the protocol tree $\cT$, the buyer's expected utility (before payment) is at least $1-\gamma$ fraction of that for $\cT$ (for our instance, the loss is at most $\gamma$). If we further trim every node that is reached with probability $\le 4^{-\beta k}$ for any buyer, the buyer's expected utility loses at most another $2^{-\beta k}$, because there are at most $2^{\beta k}$ nodes left after the first trimming step. Since we introduce new leaves after trimming, we need to specify the allocation and the payment for each of them. For each new leaf, we simply let its allocation be the empty set, and we let its payment be the least possible expected payment at this node in $\cT$ (that is, the minimum expected payment achieved by the worst possible buyer's responses in the subtree rooted at this node in $\cT$).

After the above changes, the first property obviously holds for the new protocol tree $\cT'$, and the second also holds, because if any leaf has payment larger than $4^{\beta k}$, then the probability of reaching that leaf (or node) in the original $\cT$ for any buyer is at most $4^{-\beta k}$ (otherwise the expected payment is greater than 1 for a buyer that reaches this node with probability $>4^{-\beta k}$, which exceeds the largest possible buyer's value and hence violates individual rationality), and this leaf should have been trimmed. It remains to show that the approximation factor is decreased arbitrarily little by the above changes.

To prove this, consider the buyer's (almost) dominant strategy $s^*$ in $\cT$, we change the strategy in the way that the buyer makes the same response as $s^*$ at every node that will not be trimmed by the above steps and makes the worst possible responses (which minimize the expected payment) in the subtree rooted at every node that will be trimmed. This results in a new strategy $s$ that gives the buyer almost the same expected utility as $s^*$ (as we have shown, the loss is at most $\gamma+2^{-\beta k}$). It follows that the expected payment for $s$ can only be $\gamma+2^{-\beta k}$ (plus another negligible error if the original protocol is only almost truthful) less than that for $s^*$, since otherwise the expected net utility of $s$ is significantly better than $s^*$. Moreover, the expected payment for $s$ can not be more than that for $s^*$ by definition of $s$, and hence, the expected net utility of $s$ is same as that of $s^*$ up to negligible error. Furthermore, observe that $s$ (ignore $s$'s responses at trimmed nodes) gets the same expected utility and payment for the buyer in $\cT'$. If $s$ is an almost dominant strategy in $\cT'$ (which indeed is as we will show), then we are done because we have shown the expected payment for $s$ in $\cT'$ (or $\cT$) is same (up to negligible error) as that for $s^*$ in $\cT$.

To see $s$ is an almost dominant strategy in $\cT'$, suppose for contradiction there is another strategy $s'$ with an non-negligible improvement of expected net utility over $s$ in $\cT'$. We extend $s'$ to a strategy for $\cT$ by letting it make worst possible response (which minimize the expected payment) for the nodes that will be trimmed in $\cT'$. Note that the extended $s'$ has the same expected net utility in $\cT$ as that in $\cT'$, which is significantly better than $s$'s expected net utility in $\cT'$ (and hence $s$ or $s^*$'s expected net utility in $\cT$). This contradicts that $s^*$ is a (almost) dominant strategy in $\cT$.

\subsubsection*{One protocol tree can not be shared by many priors}
Now we show the main claim that leads to the lower bound result.
\begin{claim}\label{claim:one_protocol_per_prior}
For any constant $\tau>0$ and any $m=\omega(1)$, any single protocol tree can only achieve $\tau$ approximation on $\le m$ priors in $\F$.
\end{claim}

Assume for contradiction that there are $m=w(1)$ priors $\D_{c_1},\dots,\D_{c_m}$ in $\F$ sharing the same protocol tree. By Lemma~\ref{lem:low_discrepancy}, for all but $\eta$ fraction of $x\in \X_{n,\varepsilon_1,\delta_1}$, the empirical distribution of $c_i(x)$'s for $i\in [m]$ is close to $\cR_{\ell,\varepsilon_2}$ defined in Lemma~\ref{lem:low_discrepancy}, namely, the number of $i$'s such that $c_i(x)=\varepsilon_2^{t-1}$ is $(1\pm\eta)p^{(t)}m$, where $p^{(t)}\propto \varepsilon_2^{\ell-t}$. In the rest of the proof of Claim~\ref{claim:one_protocol_per_prior}, we show that for any such $x$, the average revenue over valuations $v_{c_i}^x$ for all $i\in[m]$ achieved by the protocol tree is at most $\theta:=\frac{(\varepsilon_2+1/\ell)(1+\eta)}{1-\eta}$ fraction of the optimum. Notice that $\theta$ is a constant that we can make arbitrarily small. This will finish the proof of the claim, because for at least one of $\D_{c_1},\dots,\D_{c_m}$, the protocol tree achieves no more than the average of the expected revenues for $\D_{c_1},\dots,\D_{c_m}$, which is at most $\tau=\theta+\frac{\eta\varepsilon_2^{1-\ell}}{1-\eta}$ fraction of the optimal revenue (we generously assume that it achieves full revenue on the $\eta$ fraction of $x\in \X_{n,\varepsilon_1,\delta_1}$ that is excluded from the above analysis, and the full revenue for any $x$ from this $\eta$ fraction is at most $1$, which is at most $\varepsilon_2^{1-\ell}$ times the full revenue of any $x'$ from the other $1-\eta$ fraction), and $\frac{\eta\varepsilon_2^{1-\ell}}{1-\eta}$ is arbitrarily small by our choice of parameters.

Now consider any such $x$ that the empirical distribution of $c_i(x)$'s for $i\in [m]$ is close to $\cR_{\ell,\varepsilon_2}$, and let $C_{t}$ be the set of $c_i$'s with $c_i(x)=\varepsilon_2^{t-1}$. Without loss of generality, the buyers with valuation $v_{c_{j_t}}^x$ for all $c_{j_t}\in C_{t}$ use the same dominant strategy. Moreover, consider any $c_{j_t}\in C_t$ and any $c_{j_{t+1}}\in C_{t+1}$, we denote the expected utility and payment achieved by the prescribed dominant strategy for $v_{c_{j_{t}}}^x$ by $u_t$ and $q_t$, respectively, and analogously, we denote $u_{t+1}$ and $q_{t+1}$ for $v_{c_{j_{t+1}}}^x$. If the buyer with valuation $v_{c_{j_t}}^x$ plays the strategy for $v_{c_{j_{t+1}}}^x$ instead, he will get expected utility $u_{t+1}/\varepsilon_2$ and payment $q_{t+1}$, because by definition $v_{c_{j_{t}}}^x=v_{c_{j_{t+1}}}^x/\varepsilon_2$. By definition of (almost) dominant strategy, we have the following inequality (the inequality holds approximately when we consider almost truthful-in-expectation protocols, and the error is negligible to the later derivations)
\begin{equation}\label{eq:dominant_strategy}
    \frac{u_{t+1}}{\varepsilon_2}-q_{t+1}\le u_t-q_t.
\end{equation}
Moreover, by individual rationality,
\begin{equation}\label{eq:indvidual_rationality}
    q_{t+1}\le u_{t+1},
\end{equation}
and it follows that
\begin{align}
    q_t&\le u_t-\frac{u_{t+1}}{\varepsilon_2}+q_{t+1} && \text{(Rearranging Eq.~\eqref{eq:dominant_strategy})} \nonumber \\
    &\le u_t-\frac{u_{t+1}}{\varepsilon_2}+u_{t+1} && \text{(By Eq.~\eqref{eq:indvidual_rationality})} \nonumber \\
    &= u_t-u_{t+1}\left(\frac{1}{\varepsilon_2}-1\right). \label{eq:bound_q_t}
\end{align}
Furthermore, because $c_i(x)$'s for $i\in [m]$ are distributed like $\cR_{\ell,\varepsilon_2}$, the sum of the revenues obtained from the $v_{c_i}^x$'s for all $i\in[m]$ is at most (up to a $(1+\eta)$ multiplicative error)
\begin{align*}
    \sum_{t=1}^{\ell} mp^{(t)} q_t &\le \sum_{t=1}^{\ell-1} mp^{(t)}\left(u_t-u_{t+1}\left(\frac{1}{\varepsilon_2}-1\right)\right) + mp^{(\ell)}u_{\ell}  && \text{(By Eq.~\eqref{eq:bound_q_t} and Eq.~\eqref{eq:indvidual_rationality})} \\
    &= mp^{(1)}u_1+m\sum_{t=2}^{\ell} u_{t}\left(p^{(t)}-\frac{p^{(t-1)}}{\varepsilon_2}+p^{(t-1)}\right) && \text{(Rearranging the sum)}\\
    &= mp^{(1)}u_1+m\sum_{t=2}^{\ell} u_{t} p^{(t-1)} &&\text{(By definition of $p^{(t)}$)}\\
    &= mp^{(1)}u_1+ m\varepsilon_2\sum_{t=2}^{\ell} u_{t} p^{(t)} \\
    &\le mp^{(1)}+m\varepsilon_2\sum_{t=2}^{\ell} p^{(t)} \varepsilon_2^{t-1} &&\text{(By $u_t\le \varepsilon_2^{t-1}$)},
\end{align*}
which is at most $\varepsilon_2$ fraction of $\sum_{t=1}^{\ell} mp^{(t)} \varepsilon_2^{t-1}$ plus $mp^{(1)}$, but $mp^{(1)}$ is only $1/\ell$ fraction of $\sum_{t=1}^{\ell} mp^{(t)} \varepsilon_2^{t-1}$ by its definition. Because $c_i(x)$'s for $i\in [m]$ are distributed like $\cR_{\ell,\varepsilon_2}$, the optimal total revenue we can get from all the $v_{c_i}^x$ for $i\in[m]$ (which is equal to their total value) is at least $(1-\eta)\sum_{t=1}^{\ell} mp^{(t)} \varepsilon_2^{t-1}$, and hence, the average revenue achieved by the protocol tree on valuations $v_{c_i}^x$ for $i\in[m]$ is at most $\frac{(\varepsilon_2+1/\ell)(1+\eta)}{1-\eta}$ fraction of the optimum.

\subsubsection*{Finishing the proof by a counting argument}
For any constant $\tau>0$, suppose that the communication complexity of a $\tau$-approximate truthful-in-expectation protocol is $k=o(n)$, and without loss of generality we assume that the protocol always uses up $k$ bits. We count how many protocol trees we can have. Note that a protocol tree is determined by the $(Y,q)$ pairs on the leaves and the probabilities on the edges. Without loss of generality, we can assume that the payments and the probabilities have finite precision, namely, the probabilities are rounded to $\{i/4^n\mid i=0,1,\dots,4^n\}$, and the payments are rounded to $\{i/4^n\mid i=0,1,\dots,2^{O(n)}\}$. To see this, first observe that rounding can only change the payment at any leaf by at most $1/4^n$, and
similarly, it can only change the probability of reaching any leaf by $O(1/4^n)$, and therefore, it only changes the expected utility and the expected payment for the buyer by at most $O(2^k/4^n)=O(1/2^n)$. As we have noted along the proof, the analysis works for almost truthful-in-expectation protocols, which tolerates this extra $O(1/2^n)$ error.

Therefore, there are at most $2^n$ choices of $Y$ and at most $2^{O(n)}$ choices of $q$, which implies at most $2^{O(n)}$ choices of $(Y,q)$ at each leaf, and there are at most $4^n$ choices of the probability on each edge. Since the depth of the protocol tree is no more than $k$, there are $2^k$ leaves and $2^{k+1}$ edges at most. Altogether, there are at most $(2^{O(n)})^{2^k}\cdot(4^n)^{2^{k+1}}=2^{2^{k+o(n)}}$ possible protocol trees. Furthermore, by Claim~\ref{claim:one_protocol_per_prior}, these protocol trees can only beat $\tau$-approximation on at most $2^{2^{k+o(n)}}\cdot m$ priors in total for any $m=\omega(1)$, but there are $2^{2^{\Omega(n)}}$ priors in $\F$. Hence, most priors in $\F$ are hard for all the $o(n)$-communication protocols.

\end{proof}

\subsection{Separating the complexity of implementing and incentivizing}

\begin{remark}\label{remark:fadel_segal}
There is an $O(\log n)$-communication implementation of the optimal protocol for our hard instances. Combining with the lower bound, this shows an exponential separation between communication complexity of almost truthful-in-expectation implementation and that of non-truthful implementation for this protocol, even when the buyer's valuation has constant precision.
\end{remark}
\begin{proof}
A more communication-efficient non-truthful implementation is that the buyer randomly chooses an item $i$ of interest and sends $i$ and $c(x)$ to the seller, and then the seller gives the item $i$ to the buyer and charges the buyer $c(x)$, which only uses $O(\log n)$ bits of communication.
\end{proof}

\section{Communication lower bound for gross-substitutes valuations}
In this section, we sketch how to apply our techniques to establishing sub-exponential communication complexity lower bound for gross substitutes valuations. The specific gross substitutes valuations we use in our hard instances are the matroid rank functions provided by the following lemma.

\begin{lemma}[{\cite[Theorem 1]{BalcanHarvey18}}]\label{lem:matroid_rank}
For any $b\ge 8$ with $b=2^{o(n^{1/3})}$, there exists a family of sets $\cA\subseteq \{0,1\}^n$ and a family of matroids $\cM=\{M_{\cB}\mid \cB\subseteq\cA\}$ with the following properties:
\begin{itemize}
    \item $|\cA|=b$ and $|A|=n^{1/3}$ for every $A\in \cA$.
    \item For every $\cB\subseteq\cA$ and every $A\in\cA$, we have
    \begin{equation}
    \textnormal{rank}_{M_{\cB}}(A) =
    \begin{cases}
        |A| & A\in\cB \\
        8\log b &    A\in\cA\setminus\cB.
    \end{cases}
    \end{equation}
\end{itemize}
\end{lemma}

\begin{theorem}\label{thm:lower_bound_gross_substitutes}
For every constant $\tau>0$, any $\tau$-approximate (almost) truthful-in-expectation protocol for revenue maximization, where the seller has $n$ items, and the buyers have gross substitutes valuations, requires $2^{\Omega(n^{1/3})}$ bits of communication in expectation.
\end{theorem}

\begin{proof}
The proof follows the same strategy as the proof of Theorem~\ref{thm:lower_bound_unit_demand}.
We first construct a family of prior distributions over matroid rank valuations.

\paragraph{Construction} For any $b=2^{o(n^{1/3})}$, we let $\cA$ be the set family of size $b$ provided by Lemma~\ref{lem:matroid_rank}. For arbitrarily tiny constants $\varepsilon_1,\varepsilon_2,\delta_1,\eta>0$ and large integer constant $\ell$ such that $\eta,\varepsilon_1(1+\delta_1)\ll \varepsilon_2^{\ell}$, we take the set family $\X_{b,\varepsilon_1,\delta_1}$ by Lemma~\ref{lem:weak_design} and let $N:=|\X_{b,\varepsilon_1,\delta_1}|=2^{\Omega(b)}$, and then, we take the vector family $\C_{N,\ell,\varepsilon_2,\eta}$ from Lemma~\ref{lem:low_discrepancy} with $|\C_{N,\ell,\varepsilon_2,\eta}|=2^{\Omega(N)}=2^{2^{\Omega(b)}}$. We let each $x\in\X_{b,\varepsilon_1,\delta_1}$ represent a sub-family $\cB_x\subseteq\cA$ (each coordinate of $x$ corresponds to a distinct set $A\in\cA$, and this coordinate has value $1$ iff $A\in\cB_x$). We can fix a one-to-one mapping between the coordinates of a vector in $\C_{N,\ell,\varepsilon_2,\eta}$ and all the $x\in\X_{b,\varepsilon_1,\delta_1}$, and for any vector $c\in\C_{N,\ell,\varepsilon_2,\eta},\,x\in\X_{b,\varepsilon_1,\delta_1}$, we denote $c(x)$ as $c$'s value at the coordinate that corresponds to $x$.

For each vector $c\in\C_{N,\ell,\varepsilon_2,\eta}$, we construct a prior distribution $\D_c$ of the buyers' valuations as follows --- First, for each $x\in\X_{b,\varepsilon_1,\delta_1}$, we define a scaled matroid rank valuation $v^x_c:2^{[n]}\to \R_{\ge0}$ as
\[
    v_c^x(S) = \frac{c(x)}{n^{1/3}}\cdot \textnormal{Rank}_{M_{{\cB}_x}}(S),
\]
where $\textnormal{Rank}_{M_{{\cB}_x}}$ is the matroid rank function from Lemma~\ref{lem:matroid_rank}.
Then, we let $\D_c$ be the uniform distribution over $v_c^x$'s for all $x\in\X_{b,\varepsilon_1,\delta_1}$. Finally, the family of prior distributions is $\F=\{\D_c\mid c\in\C_{N,\ell,\varepsilon_2,\eta}\}$.

\paragraph{Interpretation}
In this instance, each $x\in\X_{b,\varepsilon_1,\delta_1}$ corresponds to a family of subsets of items that the buyer with valuation $v_c^x$ likes the most. Each vector $c\in\C_{N,\ell,\varepsilon_2,\eta}$ specifies for each $x\in\X_{b,\varepsilon_1,\delta_1}$ how valuable such a subset of items is to the buyer with valuation $v_c^x$. By the design property of $\X_{b,\varepsilon_1,\delta_1}$, every $v_c^{x_1},v_c^{x_2}$ with distinct $x_1,x_2$ are interested in mostly different subsets of items. By the design property of $\C_{N,\ell,\varepsilon_2,\eta}$, for a large number of valuations $v_c^{x}$'s with distinct $c$'s but the same $x$, the values of a subset of items in $x$ to these valuations are distributed roughly according to the ``equal revenue distribution'' $\cR_{\ell,\varepsilon_2}$ defined in Lemma~\ref{lem:low_discrepancy}.

\subsubsection*{An optimal truthful-in-expectation protocol for the hard instances}
We show that there is a truthful-in-expectation protocol that achieves optimal full revenue using $O(b)$ bits of communication for the family of Bayesian instances $\F$. The protocol is as follows: the buyer sends the $x$ corresponding to his valuation $v_c^x$ to the seller, which takes $b$ bits, and then, if $x\in\X_{b,\varepsilon_1,\delta_1}$ (otherwise the seller stops), the seller samples a set $A\in\cB_x$ uniformly at random and gives the items in $A$ to the buyer and charges him $c(x)$. The protocol is individual rational and achieves the full revenue if the buyer tells the truth. To show it is truthful in expectation, suppose the buyer sends some other $x'\in\X_{b,\varepsilon_1,\delta_1}$ (otherwise, he always gets net utility 0). Then the probability that he receives a set in $\cB_x$ is at most $\varepsilon_1(1+\delta_1)$, and moreover, if he receives a set in $\cA\setminus\cB_x$, the value he gets is $\frac{8\log b}{n^{1/3}}=o(1)$. Hence in expectation, the buyer's net utility is at most $\varepsilon_1(1+\delta_1)c(x)+o(1)-c(x')\le \varepsilon_1(1+\delta_1)+o(1)-\varepsilon_2^{\ell-1}<0$. Therefore, sending $x$ instead of $x'$ is a better strategy for the buyer. \\

The rest of the proof is same as the corresponding part of the proof of Theorem~\ref{thm:lower_bound_unit_demand}, namely, we can analogously show that any constant approximate (almost) truthful-in-expectation protocol for $\F$ requires $2^{2^{\Omega(b)}}$ distinct protocol trees, and it follows by the same counting argument that the protocol needs $\Omega(b)$ bits communication.

\end{proof}

\begin{remark}
For XOS valuations with $n$ items, the communication complexity lower bound for any constant approximation can be improved to $2^{\Omega(n)}$.
\end{remark}
\begin{proof}
The proof is basically the same as that of Theorem~\ref{thm:lower_bound_gross_substitutes}. We point out the difference in the construction of hard instances. Given a set family $\X_{n,\varepsilon_0,\delta_0}$ for arbitrarily small constants $\varepsilon_0, \delta_0>0$ from Lemma~\ref{lem:weak_design}, we let $b:=|\X_{n,\varepsilon_0,\delta_0}|=2^{\Omega(n)}$. Now consider a set family $\X_{b,\varepsilon_1,\delta_1}$ again from Lemma~\ref{lem:weak_design}. For all $x\in\X_{b,\varepsilon_1,\delta_1}$, $x$ can represent a set family $\cB_x$ of sets in $\X_{n,\varepsilon_0,\delta_0}$ (each coordinate of $x$ corresponds to a set $A\in\X_{n,\varepsilon_0,\delta_0}$, and this coordinate has value 1 iff $A\in \cB_x$). Instead of the matroid rank functions $\textnormal{Rank}_{M_{{\cB}_x}}$, here we use the binary XOS functions $\textnormal{BXOS}_{{\cB}_x}$ in our construction, which are given by
$$\textnormal{BXOS}_{{\cB}_x}(S)=\max_{A\in{\cB}_x} |A\cap S|.$$
Let $N:=|\X_{b,\varepsilon_1,\delta_1}|=2^{\Omega(b)}$, and we take a vector family $\C_{N,\ell,\varepsilon_2,\eta}$ from Lemma~\ref{lem:low_discrepancy} with $|\C_{N,\ell,\varepsilon_2,\eta}|=2^{\Omega(N)}=2^{2^{\Omega(b)}}$.
Following the notation in the proof of Theorem~\ref{thm:lower_bound_gross_substitutes}, we define the valuations as follows
$$ v_c^x(S) = c(x)\cdot \textnormal{BXOS}_{{\cB}_x}(S).$$
$\D_c$'s and $\F$ are defined as in the proof of Theorem~\ref{thm:lower_bound_gross_substitutes}.

\end{proof}

\section{Communication lower bound for XOS valuations with independent items}
In this section, we show that beating 4/5 approximation for XOS valuations with independent items requires exponential communication. (Note that constant-factor approximation is known (\eg,~\cite{RubinsteinW18}) for more general subadditive valuations with independent items).
\begin{theorem}\label{thm:lower_bound_XOS_ind}
For every constant $\tau>0$, any $(\frac{4}{5}+\tau)$-approximate (almost) truthful-in-expectation protocol for revenue maximization, where the seller has $n$ items, and the buyers have XOS valuations with independent items, requires $2^{\Omega(n)}$ bits of communication in expectation.
\end{theorem}
\begin{proof}
The proof follows the same strategy as the proof of the previous lower bounds. First, we construct a family of prior distributions of XOS valuations with independent items. We focus on the following special case of prior distributions of XOS valuations with independent items --- Given any integer $b$, for each item $i\in[n]$ there is a distribution $\D_i$ over $\R_{\ge0}^b$, an XOS valuation $v$ is generated by first sampling a vector $a^{(i)}$ from each $\D_i$ and then defined as
$$v(S)=\max_{j\in [b]} \sum_{i\in S} a^{(i)}_j.$$
In this case, $\D_1\times\D_2\times\cdots\times\D_n$ specifies a prior distribution of XOS valuations.
\paragraph{Construction}
Let $\varepsilon_0,\varepsilon_1,\delta_1,\eta,\gamma>0$ be arbitrarily tiny constants such that $\varepsilon_1(1+\delta_1)+\varepsilon_0(1+\delta_0)<\frac{1}{2-\gamma}-\frac{1}{2}$.
Given a set family $\X_{n-1,\varepsilon_0,\delta_0}$ from Lemma~\ref{lem:weak_design}, we let $b:=|\X_{n-1,\varepsilon_0,\delta_0}|=2^{\Omega(n)}$. We can think of each set in $\X_{n-1,\varepsilon_0,\delta_0}=\{x^{(1)},x^{(2)},\dots,x^{(b)}\}$ as a binary vector.
For each $i\in[n-1]$, we let $\D_i$ be the trivial distribution with singleton support $\{a^{(i)}\}$, where $a^{(i)}\in\R_{\ge0}^b$ is defined as $a^{(i)}_j=\frac{x_i^{(j)}}{(2-\gamma)\varepsilon_0(n-1)}$ for all $j\in [b]$. Now we take another set family $\X_{b,\varepsilon_1,\delta_1}=\{y^{(1)},y^{(2)},\dots,y^{(N)}\}$ from Lemma~\ref{lem:weak_design} and a vector family $\C_{N,2,\frac{1}{2},\eta}=\{c^{(1)},c^{(2)},\dots,c^{(M)}\}$ from Lemma~\ref{lem:low_discrepancy}, where $N:=|\X_{b,\varepsilon_1,\delta_1}|=2^{\Omega(b)}$ and $M:=|\C_{N,2,\frac{1}{2},\eta}|=2^{\Omega(N)}$. For each $c^{(i)}$, we let $\D_n^{c^{(i)}}$ be the uniform distribution over $\{\frac{c^{(i)}_j+1}{2}\cdot y^{(j)} \mid j\in[N]\}$. The family of prior distributions is $\F=\{\D_1\times\D_2\times\cdots\times\D_n^{c^{(i)}} \mid i\in[M]\}$. For each prior, a valuation is sampled according to the procedure described in the previous paragraph, and specifically, a valuation function $v_{c^{(i)}}^{y^{(j)}}$, determined by $c^{(i)}$ and $y^{(j)}$, is given as follows $$v_{c^{(i)}}^{y^{(j)}}(S)=\max_{t\in [b]} \mathds{1}\{n\in S\}\cdot\frac{c^{(i)}_j+1}{2}\cdot y^{(j)}_t  + \sum_{r\in S\setminus\{n\}} \frac{x_{r}^{(t)}}{(2-\gamma)\varepsilon_0(n-1)}.$$

\paragraph{Interpretation} In this instance, any valuation $v_{c^{(i)}}^{y^{(j)}}$, when restricted to items $[n-1]$, becomes a single scaled binary XOS valuation in which the clauses correspond to the scaled binary vectors $x^{(1)},x^{(2)},\dots,x^{(b)}$ (they represent pairwise nearly disjoint subsets (of items in $[n-1]$) that are equally valuable to every buyer), and each of these clause has total value $\frac{1}{2-\gamma}$. Each binary vector $y^{(j)}$ then decides which of these clauses $x^{(1)},x^{(2)},\dots,x^{(b)}$ interact with the item $n$, i.e., the item $n$ has positive contribution to the clause $x^{(t)}$ in the valuation $v_{c^{(i)}}^{y^{(j)}}$ iff $y_{t}^{(j)}=1$. (Distinct $y^{(j)}$'s define almost completely different interactions.) Each binary vector $c^{(i)}$ then specifies for each $y^{(j)}$ how large the contribution of $n$ is for each clause where, according to $y^{(j)}$, the item $n$ has positive contribution, i.e., the item $n$ contributes value $1$ to every clause $x^{(t)}$ it interacts with (i.e., for which $y_{t}^{(j)}=1$) in the valuation $v_{c^{(i)}}^{y^{(j)}}$ if $c^{(i)}_j=1$ and contributes value $0$ if otherwise. For a large number of valuations $v_{c^{(i)}}^{y^{(j)}}$'s with distinct $c^{(i)}$'s but the same $y^{(j)}$, the contributions of the item $n$ in these valuations (to every clause it interacts with) are distributed roughly according to the ``equal revenue distribution'' $\cR_{2,1/2}$.

\subsubsection*{An optimal truthful-in-expectation protocol for the hard instances}
We show that there is a truthful-in-expectation protocol that achieves optimal full revenue using $O(b)$ bits of communication for the family of Bayesian instances $\F$. The protocol is as follows: the buyer sends the $y^{(j)}$ corresponding to his valuation $v^{c^{(i)}}_{y^{(j)}}$ to the seller, which takes $b$ bits, and then, if $y^{(j)}\in\X_{b,\varepsilon_1,\delta_1}$ (otherwise the seller stops), the seller samples a set $A$ uniformly at random from the set family $\{x^{(t)} \mid t\in[b]\textrm{ s.t. } y^{(j)}_t=1\}$ and gives the items in $A\cup\{n\}$ to the buyer and charges him $\frac{c^{(i)}_j+1}{2}+\frac{1}{2-\gamma}$. It is easy to verify that the protocol is individual rational and achieves the full revenue if the buyer tells the truth. To show it is truthful in expectation, suppose the buyer sends some other $y^{(j')}\in\X_{b,\varepsilon_1,\delta_1}$ (otherwise, he always gets net utility 0). Then the probability that he receives a set in $\{x^{(t)} \mid t\in[b]\textrm{ s.t. } y^{(j)}_t=1\}$ is at most $\varepsilon_1(1+\delta_1)$ (in which case he gets value $\frac{c^{(i)}_j+1}{2}+\frac{1}{2-\gamma}$), and moreover, if he receives a set $B\in\{x^{(t)} \mid t\in[b]\textrm{ s.t. } y^{(j)}_t=0\}$, the value he gets is at most $\frac{c^{(i)}_j+1}{2}+\frac{\varepsilon_0(1+\delta_0)}{2-\gamma}$ (total value of $B\cup\{n\}$ to him). His payment is always $\frac{c^{(i)}_{j'}+1}{2}+\frac{1}{2-\gamma}$. Hence in expectation, the buyer's net utility is at most 
\begin{align*}
&\frac{c^{(i)}_j+1}{2}+\frac{\varepsilon_1(1+\delta_1)}{2-\gamma}+\frac{\varepsilon_0(1+\delta_0)}{2-\gamma}-\left(\frac{c^{(i)}_{j'}+1}{2}+\frac{1}{2-\gamma}\right) \\
&\le\frac{1}{2}-\frac{1}{2-\gamma}+\varepsilon_1(1+\delta_1)+\varepsilon_0(1+\delta_0), \quad \text{(By $c^{(i)}_j,c^{(i)}_{j'}\in\{0,1\}$ and $\gamma$ is tiny)}
\end{align*}
which is negative by our choice of $\gamma$. Therefore, sending $y^{(j)}$ instead of $y^{(j')}$ is a better strategy for the buyer.\\

Finally, we prove the following main claim. The calculation is slightly different from that of Claim~\ref{claim:one_protocol_per_prior}, but the idea is the same. Using this main claim, the proof can be finished by the same counting argument as in the proof of Theorem~\ref{thm:lower_bound_unit_demand}, which we omit here.

\begin{claim}\label{claim:one_protocol_per_prior_independent_xos}
For any constant $\tau>0$ and any $m=\omega(1)$, any single protocol tree can only achieve $\frac{4}{5}+\tau$ approximation on $\le m$ priors in $\F$.
\end{claim}
Assume for contradiction that there are $m=\omega(1)$ priors $\D_n^{c^{(i_1)}},\dots,\D_n^{c^{(i_m)}}$ (ignoring the trivial $\D_1\times\D_2\times\cdots\times\D_{n-1}$ part) in $\F$ sharing the same protocol tree. By Lemma~\ref{lem:low_discrepancy}, for all but $\eta$ fraction of $j\in[b]$, the empirical distribution of $c^{(i)}_j$'s for $i\in \{i_1,\dots,i_m\}$ is close to $\cR_{2,1/2}$ defined in Lemma~\ref{lem:low_discrepancy}, namely, the number of $i$'s such that $c^{(i)}_j=1$ is $\frac{1\pm\eta}{2}\cdot m$. In the rest of the proof of this claim, we show that for any such $j$, the average revenue over valuations $v_{c^{(i)}}^{y^{(j)}}$ for all $i\in\{i_1,\dots,i_m\}$ achieved by the protocol tree is at most $\theta=\frac{4(1+\eta)}{(4-2\gamma)(1-\eta)}\cdot\frac{4}{5}$ fraction of the optimum, and $\theta$ is a constant that we can make arbitrarily close to $\frac{4}{5}$. This will finish the proof of the claim, because for at least one of $\D_n^{c^{(i_1)}},\dots,\D_n^{c^{(i_m)}}$, the protocol tree achieves no more than the average of the expected revenues for $\D_n^{c^{(i_1)}},\dots,\D_n^{c^{(i_m)}}$, which is $\tau=\theta+\frac{\eta}{1-\eta}\cdot\frac{3}{2}$ fraction of the optimal revenue (we generously assume that it achieves full revenue on the $\eta$ fraction of $j\in[b]$ that is excluded in above analysis, and the full revenue for any $y^{(j)}$ from this $\eta$ fraction is at most $\frac{3}{2}$ times the full revenue of any $y^{(j')}$ from the other $1-\eta$ fraction), and $\frac{\eta}{1-\eta}\cdot\frac{3}{2}$ can be made arbitrarily small.

Consider any such $j$ that the empirical distribution of $c^{(i)}_j$'s for $i\in \{i_1,\dots,i_m\}$ is close to $\cR_{2,1/2}$, and let $C_{t}$ be the set of $c^{(i)}$'s such that $c^{(i)}_j=\frac{1}{2^t}$ for each $t\in\{0,1\}$. Without loss of generality, the buyers with valuation $v_{c^{(i)}_j}^{y^{(j)}}$ for all $c^{(i)}\in C_{t}$ use the same dominant strategy. Moreover, consider any $c^{(i')}\in C_0$ and any $c^{(i'')}\in C_{1}$, we denote the expected utility over items $[n-1]$, the probability of getting item $n$, and the payment achieved by the prescribed dominant strategy for $v_{c^{(i')}}^{y^{(j)}}$ by $u_0$, $p_0$ and $q_0$, respectively, and analogously, we denote $u_{1}$, $p_{1}$ and $q_{1}$ for $v_{c^{(i'')}}^{y^{(j)}}$. Notice that if the buyer with valuation $v_{c^{(i')}}^{y^{(j)}}$ plays the strategy for $v_{c^{(i'')}}^{y^{(j)}}$ instead, he will get expected utility $p_1+u_1$ (because the probability he gets item $n$ is now $p_1$, and the expected utility he gets from $[n-1]$ is now $u_1$) and payment $q_1$. By definition of (almost) dominant strategy, we have the following inequality (the inequality holds approximately when we consider almost truthful-in-expectation protocols, and the error is negligible to the later derivations)
\begin{equation}\label{eq:dominant_strategy_independent_xos}
    p_1+u_1-q_1\le p_0+u_0-q_0.
\end{equation}
Moreover, by individual rationality,
\begin{equation}\label{eq:indvidual_rationality_independent_xos}
    q_1\le \frac{p_1}{2}+u_1,
\end{equation}
and it follows that
\begin{align}
    q_0&\le p_0+u_0-u_1-p_1+q_1 && \text{(Rearranging Eq.~\eqref{eq:dominant_strategy_independent_xos})} \nonumber \\
    &\le p_0+u_0-u_1-p_1+\frac{p_1}{2}+u_1 && \text{(By Eq.~\eqref{eq:indvidual_rationality_independent_xos})} \nonumber \\
    &= p_0+u_0-\frac{p_1}{2}. \label{eq:bound_q_0_independent_xos}
\end{align}
Furthermore, because $v_{c^{(i)}}^{y^{(j)}}$'s for $i\in\{i_1,\dots,i_m\}$ are distributed like $\cR_{\ell,\varepsilon_2}$, the sum of the revenues obtained by the protocol tree from the $v_{c^{(i)}}^{y^{(j)}}$'s for all $i\in\{i_1,\dots,i_m\}$ is at most
\begin{align*}
    \frac{(1+\eta)m}{2}\cdot q_0+\frac{(1+\eta)m}{2}\cdot q_1&\le \frac{(1+\eta)m}{2}\cdot(p_0+u_0+u_1) && \text{(By Eq.~\eqref{eq:indvidual_rationality_independent_xos}, \eqref{eq:bound_q_0_independent_xos})}\\
    &\le \frac{4(1+\eta)m}{4-2\gamma}. && \text{(By $p_0\le1$ and $u_0,u_1\le \frac{1}{2-\gamma}$)}
\end{align*}
Because $v_{c^{(i)}}^{y^{(j)}}$'s for $i\in\{i_1,\dots,i_m\}$ are distributed like $\cR_{\ell,\varepsilon_2}$, the optimal total revenue we can get from all the $v_{c^{(i)}}^{y^{(j)}}$ for $i\in\{i_1,\dots,i_m\}$ (which is equal to their total value) is at least $\frac{(1-\eta)m}{2}\cdot\frac{3}{2}+\frac{(1-\eta)m}{2}=\frac{5(1-\eta)m}{4}$, and hence, the average revenue achieved by the protocol tree on valuations $v_{c^{(i)}}^{y^{(j)}}$ for $i\in\{i_1,\dots,i_m\}$ is at most $\frac{4(1+\eta)}{(4-2\gamma)(1-\eta)}\cdot\frac{4}{5}$ fraction of the optimum.
\end{proof}

\bibliographystyle{alpha}
\bibliography{MasterBib}

\newcommand{\etalchar}[1]{$^{#1}$}
\begin{thebibliography}{CGMW18}

\bibitem[AKSW20]{AssadiKSW20}
Sepehr Assadi, Hrishikesh Khandeparkar, Raghuvansh~R Saxena, and S~Matthew
  Weinberg.
\newblock Separating the communication complexity of truthful and non-truthful
  combinatorial auctions.
\newblock In {\em Proceedings of the 52nd Annual ACM SIGACT Symposium on Theory
  of Computing}, pages 1073--1085, 2020.

\bibitem[APTT03]{ArcherPTT03}
Aaron Archer, Christos~H. Papadimitriou, Kunal Talwar, and {\'{E}}va Tardos.
\newblock An approximate truthful mechanism for combinatorial auctions with
  single parameter agents.
\newblock {\em Internet Math.}, 1(2):129--150, 2003.

\bibitem[Ass17]{Assadi17}
Sepehr Assadi.
\newblock Combinatorial auctions do need modest interaction.
\newblock In {\em Proceedings of the 2017 {ACM} Conference on Economics and
  Computation, {EC} '17, Cambridge, MA, USA, June 26-30, 2017}, pages 145--162,
  2017.

\bibitem[Bar38]{Barnard38}
Chester~Irving Barnard.
\newblock {\em The Functions of the Executive}.
\newblock Harvard University Press, 1938.

\bibitem[BBS13]{BabaioffBS13}
Moshe Babaioff, Liad Blumrosen, and Michael Schapira.
\newblock The communication burden of payment determination.
\newblock {\em Games Econ. Behav.}, 77(1):153--167, 2013.

\bibitem[BCKW10]{BriestCKW10}
Patrick Briest, Shuchi Chawla, Robert Kleinberg, and S.~Matthew Weinberg.
\newblock Pricing randomized allocations.
\newblock In {\em Proceedings of the Twenty-First Annual {ACM-SIAM} Symposium
  on Discrete Algorithms, {SODA} 2010, Austin, Texas, USA, January 17-19,
  2010}, pages 585--597, 2010.

\bibitem[BF13]{BlumrosenF13}
Liad Blumrosen and Michal Feldman.
\newblock Mechanism design with a restricted action space.
\newblock {\em Games Econ. Behav.}, 82:424--443, 2013.

\bibitem[BGN17]{BabaioffGN17}
Moshe Babaioff, Yannai~A. Gonczarowski, and Noam Nisan.
\newblock The menu-size complexity of revenue approximation.
\newblock In {\em Proceedings of the 49th Annual {ACM} {SIGACT} Symposium on
  Theory of Computing, {STOC} 2017, Montreal, QC, Canada, June 19-23, 2017},
  pages 869--877, 2017.

\bibitem[BH18]{BalcanHarvey18}
Maria-Florina Balcan and Nicholas~JA Harvey.
\newblock Submodular functions: Learnability, structure, and optimization.
\newblock {\em SIAM Journal on Computing}, 47(3):703--754, 2018.

\bibitem[BILW14]{BabaioffILW14}
Moshe Babaioff, Nicole Immorlica, Brendan Lucier, and S.~Matthew Weinberg.
\newblock A simple and approximately optimal mechanism for an additive buyer.
\newblock In {\em 55th {IEEE} Annual Symposium on Foundations of Computer
  Science, {FOCS} 2014, Philadelphia, PA, USA, October 18-21, 2014}, pages
  21--30, 2014.

\bibitem[BKS15]{BabaioffKS15}
Moshe Babaioff, Robert~D. Kleinberg, and Aleksandrs Slivkins.
\newblock Truthful mechanisms with implicit payment computation.
\newblock {\em J. {ACM}}, 62(2):10:1--10:37, 2015.

\bibitem[BMW16]{BravermanMW16}
Mark Braverman, Jieming Mao, and S.~Matthew Weinberg.
\newblock Interpolating between truthful and non-truthful mechanisms for
  combinatorial auctions.
\newblock In {\em Proceedings of the Twenty-Seventh Annual {ACM-SIAM} Symposium
  on Discrete Algorithms, {SODA} 2016, Arlington, VA, USA, January 10-12,
  2016}, pages 1444--1457, 2016.

\bibitem[BMW18]{BravermanMW18}
Mark Braverman, Jieming Mao, and S.~Matthew Weinberg.
\newblock On simultaneous two-player combinatorial auctions.
\newblock In {\em Proceedings of the Twenty-Ninth Annual {ACM-SIAM} Symposium
  on Discrete Algorithms, {SODA} 2018, New Orleans, LA, USA, January 7-10,
  2018}, pages 2256--2273, 2018.

\bibitem[BN19]{BranzeiN19}
Simina Br{\^{a}}nzei and Noam Nisan.
\newblock Communication complexity of cake cutting.
\newblock In {\em Proceedings of the 2019 {ACM} Conference on Economics and
  Computation, {EC} 2019, Phoenix, AZ, USA, June 24-28, 2019.}, page 525, 2019.

\bibitem[BNS07]{BlumrosenNS07}
Liad Blumrosen, Noam Nisan, and Ilya Segal.
\newblock Auctions with severely bounded communication.
\newblock {\em J. Artif. Intell. Res.}, 28:233--266, 2007.

\bibitem[BR17]{BabichenkoR17}
Yakov Babichenko and Aviad Rubinstein.
\newblock Communication complexity of approximate nash equilibria.
\newblock In {\em Proceedings of the 49th Annual {ACM} {SIGACT} Symposium on
  Theory of Computing, {STOC} 2017, Montreal, QC, Canada, June 19-23, 2017},
  pages 878--889, 2017.

\bibitem[BW19]{BeyhaghiW19}
Hedyeh Beyhaghi and S.~Matthew Weinberg.
\newblock Optimal (and benchmark-optimal) competition complexity for additive
  buyers over independent items.
\newblock In {\em Proceedings of the 51st ACM Symposium on Theory of Computing
  Conference (STOC)}, 2019.

\bibitem[CDO{\etalchar{+}}15]{ChenDOPSY15}
Xi~Chen, Ilias Diakonikolas, Anthi Orfanou, Dimitris Paparas, Xiaorui Sun, and
  Mihalis Yannakakis.
\newblock On the complexity of optimal lottery pricing and randomized
  mechanisms.
\newblock In {\em {IEEE} 56th Annual Symposium on Foundations of Computer
  Science, {FOCS} 2015, Berkeley, CA, USA, 17-20 October, 2015}, pages
  1464--1479, 2015.

\bibitem[CDP{\etalchar{+}}14]{ChenDPSY14}
Xi~Chen, Ilias Diakonikolas, Dimitris Paparas, Xiaorui Sun, and Mihalis
  Yannakakis.
\newblock The complexity of optimal multidimensional pricing.
\newblock In {\em Proceedings of the Twenty-Fifth Annual {ACM-SIAM} Symposium
  on Discrete Algorithms, {SODA} 2014, Portland, Oregon, USA, January 5-7,
  2014}, pages 1319--1328, 2014.

\bibitem[CDW16]{CaiDW16}
Yang Cai, Nikhil Devanur, and S.~Matthew Weinberg.
\newblock A duality based unified approach to bayesian mechanism design.
\newblock In {\em Proceedings of the 48th ACM Conference on Theory of
  Computation(STOC)}, 2016.

\bibitem[CGMW18]{ChengGMW18}
Yu~Cheng, Nick Gravin, Kamesh Munagala, and Kangning Wang.
\newblock A simple mechanism for a budget-constrained buyer.
\newblock In {\em Web and Internet Economics - 14th International Conference,
  {WINE} 2018, Oxford, UK, December 15-17, 2018, Proceedings}, pages 96--110,
  2018.

\bibitem[CHK07]{ChawlaHK07}
Shuchi Chawla, Jason~D. Hartline, and Robert~D. Kleinberg.
\newblock {Algorithmic Pricing via Virtual Valuations}.
\newblock In {\em the 8th ACM Conference on Electronic Commerce (EC)}, 2007.

\bibitem[CHMS10]{ChawlaHMS10}
Shuchi Chawla, Jason~D. Hartline, David~L. Malec, and Balasubramanian Sivan.
\newblock {Multi-Parameter Mechanism Design and Sequential Posted Pricing}.
\newblock In {\em the 42nd ACM Symposium on Theory of Computing (STOC)}, 2010.

\bibitem[CM16]{ChawlaM16}
Shuchi Chawla and J.~Benjamin Miller.
\newblock Mechanism design for subadditive agents via an ex ante relaxation.
\newblock In {\em Proceedings of the 2016 {ACM} Conference on Economics and
  Computation, {EC} '16, Maastricht, The Netherlands, July 24-28, 2016}, pages
  579--596, 2016.

\bibitem[CMS15]{ChawlaMS15}
Shuchi Chawla, David~L. Malec, and Balasubramanian Sivan.
\newblock The power of randomness in bayesian optimal mechanism design.
\newblock {\em Games and Economic Behavior}, 91:297--317, 2015.

\bibitem[CP11]{CaragiannisP11}
Ioannis Caragiannis and Ariel~D. Procaccia.
\newblock Voting almost maximizes social welfare despite limited communication.
\newblock {\em Artif. Intell.}, 175(9-10):1655--1671, 2011.

\bibitem[CS04a]{ConitzerS04}
Vincent Conitzer and Tuomas Sandholm.
\newblock Communication complexity as a lower bound for learning in games.
\newblock In {\em Proceedings of the twenty-first international conference on
  Machine learning}, page~24. ACM, 2004.

\bibitem[CS04b]{ConitzerS04b}
Vincent Conitzer and Tuomas Sandholm.
\newblock Computational criticisms of the revelation principle.
\newblock In {\em Proceedings 5th {ACM} Conference on Electronic Commerce
  (EC-2004), New York, NY, USA, May 17-20, 2004}, pages 262--263, 2004.

\bibitem[CS05]{ConitzerS05}
Vincent Conitzer and Tuomas Sandholm.
\newblock Communication complexity of common voting rules.
\newblock In {\em Proceedings 6th {ACM} Conference on Electronic Commerce
  (EC-2005), Vancouver, BC, Canada, June 5-8, 2005}, pages 78--87, 2005.

\bibitem[CTT19]{ChawlaTT19}
Shuchi Chawla, Yifeng Teng, and Christos Tzamos.
\newblock Buy-many mechanisms are not much better than item pricing.
\newblock In {\em Proceedings of the 2019 {ACM} Conference on Economics and
  Computation, {EC} 2019, Phoenix, AZ, USA, June 24-28, 2019.}, pages 237--238,
  2019.

\bibitem[CZ17]{CaiZ17}
Yang Cai and Mingfei Zhao.
\newblock Simple mechanisms for subadditive buyers via duality.
\newblock In {\em Proceedings of the 49th Annual {ACM} {SIGACT} Symposium on
  Theory of Computing, {STOC} 2017, Montreal, QC, Canada, June 19-23, 2017},
  pages 170--183, 2017.

\bibitem[DD13]{DobzinskiD13}
Shahar Dobzinski and Shaddin Dughmi.
\newblock On the power of randomization in algorithmic mechanism design.
\newblock {\em {SIAM} J. Comput.}, 42(6):2287--2304, 2013.

\bibitem[DDT14]{DaskalakisDT14}
Constantinos Daskalakis, Alan Deckelbaum, and Christos Tzamos.
\newblock {The Complexity of Optimal Mechanism Design}.
\newblock In {\em the 25th ACM-SIAM Symposium on Discrete Algorithms (SODA)},
  2014.

\bibitem[DDT17]{DaskalakisDT17}
Constantinos Daskalakis, Alan Deckelbaum, and Christos Tzamos.
\newblock Strong duality for a multiple-good monopolist.
\newblock {\em Econometrica}, 85(3):735--767, 2017.

\bibitem[DNO14]{DobzinskiNO14}
Shahar Dobzinski, Noam Nisan, and Sigal Oren.
\newblock Economic efficiency requires interaction.
\newblock In {\em the 46th annual ACM symposium on Theory of computing (STOC)},
  2014.

\bibitem[Dob16a]{Dobzinski16a}
Shahar Dobzinski.
\newblock Breaking the logarithmic barrier for truthful combinatorial auctions
  with submodular bidders.
\newblock In {\em Proceedings of the 48th Annual ACM SIGACT Symposium on Theory
  of Computing}, STOC 2016, pages 940--948, New York, NY, USA, 2016. ACM.

\bibitem[Dob16b]{Dobzinski16b}
Shahar Dobzinski.
\newblock Computational efficiency requires simple taxation.
\newblock In {\em FOCS}, 2016.

\bibitem[DR20]{DR21}
Shahar Dobzinski and Shiri Ron.
\newblock The communication complexity of payment computation.
\newblock {\em CoRR}, abs/2012.14623, 2020.

\bibitem[EFF{\etalchar{+}}17]{EdenFFTW17b}
Alon Eden, Michal Feldman, Ophir Friedler, Inbal Talgam{-}Cohen, and S.~Matthew
  Weinberg.
\newblock The competition complexity of auctions: {A} bulow-klemperer result
  for multi-dimensional bidders.
\newblock In {\em Proceedings of the 2017 {ACM} Conference on Economics and
  Computation, {EC} '17, Cambridge, MA, USA, June 26-30, 2017}, page 343, 2017.

\bibitem[EFN{\etalchar{+}}19]{EzraFNTW19}
Tomer Ezra, Michal Feldman, Eric Neyman, Inbal Talgam-Cohen, and S.~Matthew
  Weinberg.
\newblock Settling the communication complexity of combinatorial auctions with
  two subadditive buyers.
\newblock In {\em the 60th Annual IEEE Symposium on Foundations of Computer
  Science (FOCS)}, 2019.

\bibitem[FFR18]{FeldmanFR18}
Michal Feldman, Ophir Friedler, and Aviad Rubinstein.
\newblock 99{\%} revenue via enhanced competition.
\newblock In {\em Proceedings of the 2018 {ACM} Conference on Economics and
  Computation, Ithaca, NY, USA, June 18-22, 2018}, pages 443--460, 2018.

\bibitem[FS09]{FadelS09}
Ronald Fadel and Ilya Segal.
\newblock The communication cost of selfishness.
\newblock {\em J. Econ. Theory}, 144(5):1895--1920, 2009.

\bibitem[GK18]{GanorS18}
Anat Ganor and {Karthik {C. S.}}
\newblock Communication complexity of correlated equilibrium with small
  support.
\newblock In {\em Approximation, Randomization, and Combinatorial Optimization.
  Algorithms and Techniques, {APPROX/RANDOM} 2018, August 20-22, 2018 -
  Princeton, NJ, {USA}}, pages 12:1--12:16, 2018.

\bibitem[GKP19]{GanorKP19}
Anat Ganor, {Karthik {C. S.}}, and D{\"{o}}m{\"{o}}t{\"{o}}r
  P{\'{a}}lv{\"{o}}lgyi.
\newblock On communication complexity of fixed point computation.
\newblock {\em CoRR}, abs/1909.10958, 2019.

\bibitem[GNOR19]{GonczarowskiNOR19}
Yannai~A. Gonczarowski, Noam Nisan, Rafail Ostrovsky, and Will Rosenbaum.
\newblock A stable marriage requires communication.
\newblock {\em Games Econ. Behav.}, 118:626--647, 2019.

\bibitem[Gon18]{Gonczarowski18}
Yannai~A. Gonczarowski.
\newblock Bounding the menu-size of approximately optimal auctions via
  optimal-transport duality.
\newblock In {\em Proceedings of the 50th Annual {ACM} {SIGACT} Symposium on
  Theory of Computing, {STOC} 2018, Los Angeles, CA, USA, June 25-29, 2018},
  pages 123--131, 2018.

\bibitem[GR18]{GoosR18}
Mika G{\"{o}}{\"{o}}s and Aviad Rubinstein.
\newblock Near-optimal communication lower bounds for approximate nash
  equilibria.
\newblock In {\em 59th {IEEE} Annual Symposium on Foundations of Computer
  Science, {FOCS} 2018, Paris, France, October 7-9, 2018}, pages 397--403,
  2018.

\bibitem[Hay45]{Hayek45}
F.~A. Hayek.
\newblock The use of knowledge in society.
\newblock {\em The American Economic Review}, 35(4):519--530, 1945.

\bibitem[HM10]{HartM10}
Sergiu Hart and Yishay Mansour.
\newblock How long to equilibrium? the communication complexity of uncoupled
  equilibrium procedures.
\newblock {\em Games and Economic Behavior}, 69(1):107--126, 2010.

\bibitem[HN17]{HartN17}
Sergiu Hart and Noam Nisan.
\newblock Approximate revenue maximization with multiple items.
\newblock {\em J. Economic Theory}, 172:313--347, 2017.

\bibitem[HN19]{HartN13}
Sergiu Hart and Noam Nisan.
\newblock Selling multiple correlated goods: Revenue maximization and menu-size
  complexity.
\newblock {\em J. Economic Theory}, 183:991--1029, 2019.

\bibitem[HR09]{HartlineR09}
Jason~D. Hartline and Tim Roughgarden.
\newblock Simple versus optimal mechanisms.
\newblock In {\em ACM Conference on Electronic Commerce}, pages 225--234, 2009.

\bibitem[HR15]{HartR15}
Sergiu Hart and Philip~J. Reny.
\newblock {Maximizing Revenue with Multiple Goods: Nonmonotonicity and Other
  Observations}.
\newblock {\em Theoretical Economics}, 10(3):893--922, 2015.

\bibitem[HR17]{HartR17}
Sergiu Hart and Philip~J. Reny.
\newblock The better half of selling separately.
\newblock {\em CoRR}, abs/1712.08973, 2017.

\bibitem[Hur60]{Hurwicz60}
L.~Hurwicz.
\newblock {\em Optimality and Informational Efficiency in Resource Allocation
  Processes}.
\newblock Reprint series. Stanford University Press, 1960.

\bibitem[ILM{\etalchar{+}}18]{ImmorlicaLMST18}
Nicole Immorlica, Brendan Lucier, Jieming Mao, Vasilis Syrgkanis, and Christos
  Tzamos.
\newblock Combinatorial assortment optimization.
\newblock In George Christodoulou and Tobias Harks, editors, {\em Web and
  Internet Economics - 14th International Conference, {WINE} 2018, Oxford, UK,
  December 15-17, 2018, Proceedings}, volume 11316 of {\em Lecture Notes in
  Computer Science}, pages 218--231. Springer, 2018.

\bibitem[KMS{\etalchar{+}}19]{KothariMSSW19}
Pravesh Kothari, Divyarthi Mohan, Ariel Schvartzman, Sahil Singla, and
  S.~Matthew Weinberg.
\newblock Approximation schemes for a buyer with independent items via
  symmetries.
\newblock In {\em the 60th Annual IEEE Symposium on Foundations of Computer
  Science (FOCS)}, 2019.

\bibitem[KW12]{KleinbergW12}
Robert Kleinberg and S.~Matthew Weinberg.
\newblock Matroid prophet inequalities.
\newblock In {\em Proceedings of the 44th Symposium on Theory of Computing
  Conference, {STOC} 2012, New York, NY, USA, May 19 - 22, 2012}, pages
  123--136, 2012.

\bibitem[LLN06]{LLN06}
Benny Lehmann, Daniel Lehmann, and Noam Nisan.
\newblock Combinatorial auctions with decreasing marginal utilities.
\newblock {\em Games Econ. Behav.}, 55(2):270--296, 2006.

\bibitem[LP18]{LiuP18}
Siqi Liu and Christos{-}Alexandros Psomas.
\newblock On the competition complexity of dynamic mechanism design.
\newblock In {\em Proceedings of the Twenty-Ninth Annual {ACM-SIAM} Symposium
  on Discrete Algorithms, {SODA} 2018, New Orleans, LA, USA, January 7-10,
  2018}, pages 2008--2025, 2018.

\bibitem[LSZ11]{LSZ11}
Hagay Levin, Michael Schapira, and Aviv Zohar.
\newblock Interdomain routing and games.
\newblock {\em {SIAM} J. Comput.}, 40(6):1892--1912, 2011.

\bibitem[LY13]{LiY13}
Xinye Li and Andrew Chi-Chih Yao.
\newblock On revenue maximization for selling multiple independently
  distributed items.
\newblock {\em Proceedings of the National Academy of Sciences},
  110(28):11232--11237, 2013.

\bibitem[MR74]{MountR74}
Kenneth Mount and Stanley Reiter.
\newblock The informational size of message spaces.
\newblock {\em Journal of Economic Theory}, 8(2):161 -- 192, 1974.

\bibitem[MV10]{ManelliV10}
A.~M. Manelli and D.~R. Vincent.
\newblock {Bayesian and Dominant-Strrategy Implementation in the Independent
  Private-Values Model}.
\newblock {\em Econometrica}, 78(6):1905--1938, 2010.

\bibitem[NS06]{NisanS06}
Noam Nisan and Ilya Segal.
\newblock The communication requirements of efficient allocations and
  supporting prices.
\newblock {\em J. Economic Theory}, 129(1):192--224, 2006.

\bibitem[Pav11]{Pavlov11}
Gregory Pavlov.
\newblock Optimal mechanism for selling two goods.
\newblock {\em The B.E. Journal of Theoretical Economics}, 11(3), 2011.

\bibitem[PR06]{ProcacciaR06a}
Ariel~D. Procaccia and Jeffrey~S. Rosenschein.
\newblock The communication complexity of coalition formation among autonomous
  agents.
\newblock In {\em 5th International Joint Conference on Autonomous Agents and
  Multiagent Systems {(AAMAS} 2006), Hakodate, Japan, May 8-12, 2006}, pages
  505--512, 2006.

\bibitem[PR19]{PlautR19}
Benjamin Plaut and Tim Roughgarden.
\newblock Communication complexity of discrete fair division.
\newblock In {\em Proceedings of the Thirtieth Annual {ACM-SIAM} Symposium on
  Discrete Algorithms, {SODA} 2019, San Diego, California, USA, January 6-9,
  2019}, pages 2014--2033, 2019.

\bibitem[PSS08]{PapadimitriouSS08}
Christos~H. Papadimitriou, Michael Schapira, and Yaron Singer.
\newblock On the hardness of being truthful.
\newblock In {\em Proceedings of the 49th Annual IEEE Symposium on Foundations
  of Computer Science (FOCS)}, 2008.

\bibitem[Rei84]{Reichelstein84}
Stefan Reichelstein.
\newblock Incentive compatibility and informational requirements.
\newblock {\em Journal of Economic Theory}, 34(1):32 -- 51, 1984.

\bibitem[RST{\etalchar{+}}20]{RSTWZ21}
Aviad Rubinstein, Raghuvansh~R. Saxena, Clayton Thomas, S.~Matthew Weinberg,
  and Junyao Zhao.
\newblock Exponential communication separations between notions of selfishness.
\newblock {\em CoRR}, abs/2012.14898, 2020.

\bibitem[RTCY12]{RoughgardenTY12}
Tim Roughgarden, Inbal Talgam-Cohen, and Qiqi Yan.
\newblock Supply-limiting mechanisms.
\newblock In {\em 13th ACM Conference on Electronic Commerce (EC)}, 2012.

\bibitem[Rub16]{Rubinstein16a}
Aviad Rubinstein.
\newblock On the computational complexity of optimal simple mechanisms.
\newblock In {\em Proceedings of the 2016 {ACM} Conference on Innovations in
  Theoretical Computer Science, Cambridge, MA, USA, January 14-16, 2016}, pages
  21--28, 2016.

\bibitem[RW16]{RoghgardenW16}
Tim Roughgarden and Omri Weinstein.
\newblock On the communication complexity of approximate fixed points.
\newblock In {\em Electronic Colloquium on Computational Complexity (ECCC)},
  volume~23, page~55, 2016.

\bibitem[RW18]{RubinsteinW18}
Aviad Rubinstein and S.~Matthew Weinberg.
\newblock Simple mechanisms for a subadditive buyer and applications to revenue
  monotonicity.
\newblock {\em {ACM} Trans. Economics and Comput.}, 6(3-4):19:1--19:25, 2018.

\bibitem[SA12]{ServiceA12a}
Travis~C. Service and Julie~A. Adams.
\newblock Communication complexity of approximating voting rules.
\newblock In {\em International Conference on Autonomous Agents and Multiagent
  Systems, {AAMAS} 2012, Valencia, Spain, June 4-8, 2012 {(3} Volumes)}, pages
  593--602, 2012.

\bibitem[Sch03]{Sch03}
Gideon Schechtman.
\newblock Chapter 37 - concentration, results and applications.
\newblock volume~2 of {\em Handbook of the Geometry of Banach Spaces}, pages
  1603 -- 1634. Elsevier Science B.V., 2003.

\bibitem[Seg07]{Segal07}
Ilya Segal.
\newblock The communication requirements of social choice rules and supporting
  budget sets.
\newblock {\em J. Econ. Theory}, 136(1):341--378, 2007.

\bibitem[Seg10]{Segal10}
Ilya~R. Segal.
\newblock Nash implementation with little communication.
\newblock {\em Theoretical Economics}, 5(1):51--71, 2010.

\bibitem[SSW18]{SaxenaSW18}
Raghuvansh~R. Saxena, Ariel Schvartzman, and S.~Matthew Weinberg.
\newblock The menu complexity of "one-and-a-half-dimensional" mechanism design.
\newblock In {\em Proceedings of the Twenty-Ninth Annual {ACM-SIAM} Symposium
  on Discrete Algorithms, {SODA} 2018, New Orleans, LA, USA, January 7-10,
  2018}, pages 2026--2035, 2018.

\bibitem[Tha04]{Thanassoulis04}
John Thanassoulis.
\newblock Haggling over substitutes.
\newblock {\em Journal of Economic Theory}, 117:217--245, 2004.

\bibitem[Wei20]{Weinberg20}
S.~Matthew Weinberg, 2020.
\newblock Personal communication.

\end{thebibliography}

\appendix

\section{A more efficient non-IC protocol\label{sec:non-IC}}

In this appendix we explore the complexity of non-IC auction protocols
that implement direct revelation (IC) mechanisms, i.e. whenever the
Buyer follows the suggested strategy, the expected allocation and
payment are still identical to those in the direct revelation mechanism.
Yet due to the order in which the information is revealed to the Buyer,
he has incentives to deviate from this suggested strategy. This corresponds
to \cite{FadelS09}'s notion of efficiently {\em implementable} mechanism
(as opposed to the stronger {\em incentivizable} mechanism).
\begin{thm}
\label{thm:non-truthful}Let $\mathcal{\ensuremath{D}}$ be any prior
over Buyer's combinatorial valuations over $n$ items bounded by maximum
valuation $U$, and let $\mathcal{\ensuremath{M}}$ be any direct
revelation mechanism. Suppose that for any type and realization of
randomness, $\mathcal{M}$ only ever allocates one of $B$ bundles.
Then there exists a non-IC protocol implementing $\mathcal{M}$ using
$O(\log(B))$ bits of communication.
\end{thm}

\begin{proof}
Similarly to the proof Theorem \ref{thm:truthful-B}, we consider
a partition of $[0,1]$ into $B$ intervals, where the $b$-th interval
is of length identical to the probability that $\mathcal{\ensuremath{M}}$
allocates Bundle $b$ to the Buyer. Also, as in the proofs of Theorems
\ref{thm:additive} and \ref{thm:truthful-B}, we transform $\mathcal{\ensuremath{M}}$
into a mechanism $\mathcal{\ensuremath{M}}^{'}$ with payment $0$
or $U$. Notice that $\mathcal{\ensuremath{M}}^{'}$ has the same
expected payment and allocation as $\mathcal{\ensuremath{M}}$, so
implementing $\mathcal{\ensuremath{M}}^{'}$ is equivalent to implementing
$\mathcal{\ensuremath{M}}$.

Like our IC auction protocols (Theorems \ref{thm:additive} and \ref{thm:truthful-B}),
the nodes of Chance are parameterized by a uniformly random number
$\tau\in[0,1]$. We deviate from those IC auction protocols and reveal
(prefixes of) $\tau$ to the Buyer as soon as possible. Specifically,
the nodes of Chance iteratively draw the bits in the binary representation
of $\tau$. At each iteration the Buyer can terminate the protocol
or ask to reveal $\tau$ to greater precision. He terminates the protocol
iff he can determine both: (i) whether $\tau$ is greater than the
probability he has to pay $U$, and (ii) which interval contains $\tau$;
in this case he simply announces his ex-post payment and allocation
($1$ bit for the former, $\lceil\log(B)\rceil$ bits for the latter). 

Note that if the Buyer follows the suggested strategy, the distributions
of allocation and payment are identical to $\mathcal{\ensuremath{M}}^{'}$,
so this protocol indeed implements $\mathcal{\ensuremath{M}}^{'}$. 

The analysis of the communication complexity of the protocol is similar
to that of Theorems \ref{thm:additive} and \ref{thm:truthful-B}.
The Buyer must be able to determine whether $\tau$ is greater or
smaller than $B$ numbers (the probability of payment and $B-1$ interval
boundaries). He can do this once the prefix of $\tau$ he received
is different than the corresponding $B$ prefixes. Since $\tau$ is
a uniformly random number its prefix has probability $1/2$ of deviating
from each of the $B$ numbers at each iteration (independently across
iterations but not across numbers). After $2\log(B)$ iterations,
the Buyer can terminate the protocol with probability $\geq1-1/B^{2}$.
Therefore the number of iterations is $O(\log(B))$ (see also Ineq.
(\ref{eq:recursive})); the last iteration uses $\log(B)+O(1)$ bits,
and any iteration requires before that uses two bits. So the total
communication complexity is $O(\log(B))$.
\end{proof}
To appreciate why the above mechanism is non-truthful, we provide
a simple example:
\begin{example}
[IC vs non-IC auction protocols] \label{ex:non-truthful}

Consider the following mechanism $\mathcal{M}$ for auctioning a single
item: the Buyer pays $1$ with probability $p$ (otherwise zero);
he receives the item with probability $q$; and he can choose any
$p,q\in[0,1]$ such that $p=q^{2}$. Consider a Buyer type that favors
option $q=2/3,p=4/9$. In the context of our non-IC auction protocol,
this induces the partition of $[0,1]$ into $[0,2/3]\cup(2/3,1]$
(where the buyer receives the item if $\tau$ is in the first part).
If the Buyer learns that the first bit of $\tau$ is zero (i.e. $\tau<1/2$),
the Buyer would prefer to report type $q'=1/2,p'=1/4$ over his true
type since he is guaranteed to receive the item anyway.
\end{example}

\section{Special case: a protocol for the \cite{DaskalakisDT17} example\label{app:DDT}}

In this appendix we prove the concrete (non-asymptotic) bound on the
expected number of bits that the buyer sends in the DDT example. 
\begin{thm*}
(Theorem \ref{thm:DDT} restated) Consider the case of $n=2$ items
and the Buyer drawing his valuations i.i.d. from $\Beta(1,2)$ (i.e.~the
distribution on $[0,1]$ with density function $f(x)=2(1-x)$). Then
there is an IC auction protocol obtaining the maximum possible revenue
where the Buyer sends less than two bits in expectation.
\end{thm*}
\begin{proof}
The optimal direct revelation mechanism for this distribution is analyzed
in detail in \cite[Example 3]{DaskalakisDT17}. We summarize the properties
useful for our proof: The optimal mechanism partitions the Buyer's
types into four regions $\mathcal{Z},\mathcal{A},\mathcal{B},\mathcal{W}$.
Buyer types in $\mathcal{Z}$ pay zero and never receive either item.
Buyer types in $\mathcal{W}$ pay a fixed price ($P\approx0.5535$)
and receive both items. Buyer types in $\mathcal{A}$ and $\mathcal{B}$
have a strict preference between the items (different preferences
between $\mathcal{A}$ and $\mathcal{B}$); they always receive their
most preferred item, and receive their least preferred item with probability
$\pi$ which is always in the range $\pi\in[1/8,1/8+0.03)$. Their
payment is always less than $P$. (The exact payment and probability of
allocation is slightly different for each type in $\mathcal{A}\cup\mathcal{B}$.)

\paragraph*{A protocol for $\mathcal{A}\cup\mathcal{B}$}

We now describe a simple IC auction protocol for the special case
where the Buyer's type is in $\mathcal{A}\cup\mathcal{B}$; in this
mechanism the Buyer sends less than one bit in expectation. We will
then show how to use this mechanism to prove the theorem for the general
case. We first convert the original mechanism to one where the Buyer
always pays either zero or $U$, for some arbitrarily large constant
$U$; the payment of $U$ is charged with probability $q<1/U$ (the
exact probability depends on the Buyer's type).

The nodes of Chance are parameterized by a threshold $\tau$ drawn
uniformly at random from $[0,1]$. At the end of the protocol, the
Buyer should always receive the more desired item, receive the less
desired item iff $\tau<\pi$, and pay $U$ iff $\tau<q$. We now consider
the following protocols, depending on $\tau$ (see summary in Table
\ref{tab:A-cup-B-1-1}):
\begin{itemize}
\item If $1/U<\tau<1/8$ (which happens with probability arbitrarily close
to $1/8$), the Buyer always receives both items and pay nothing.
In this case the protocol can terminate with zero communication from
the Buyer. 
\item If $\tau>1/8+0.03$, the Buyer that will receive one item and pay
nothing. The buyer sends one bit to choose which item he prefers.
\item Otherwise, i.e.~if $\tau\in[1/8,1/8+0.03)$ or $\tau<1/U$, the Buyer
needs to send more refined information about his valuation. First,
the Buyer sends his preferred item (one bit of information). Then,
the Buyer enters a sub-protocol similar to the proof of Theorem \ref{thm:additive}.
The sub-protocol is identical to both cases as the Buyer must not
learn any more information about $\tau$.
\begin{itemize}
\item If $\tau\in[1/8,1/8+0.03)$, the Buyer's suggested strategy is to
send at each round the next bit in the binary representations of $\frac{\pi-1/8}{0.03}\in[0,1)$.
Every $\sqrt{U}$ rounds, the Buyer should also send the next bit
in the binary representation of $qU\in[0,1)$. The protocol terminates
once the Buyer sent enough information to determine whether $\pi>\tau$.
Conditioning on being in this case, this happens with probability
$1/2$ at each iteration, so this part of the protocol lasts $2$
rounds in expectation. In total the Buyer expects to send just over
$3$ bits in this case: $1$ for choosing the item, $2$ bits in expectation
from the prefix of $\frac{\pi-1/8}{0.03}$, and $O(2^{-\sqrt{U}})$
bits in expectation from the prefix of $qU$.
\item If $\tau<1/U$, the Buyer's suggested strategy is as above, but the
protocol terminates once it learns whether $q>\tau$. In this case
the buyer must waste $\sqrt{U}$ bits for every bit from the prefix
of $qU$ that he actually sends. His total expected communication
$1+2(\sqrt{U}+1)=3+2\sqrt{U}$.
\end{itemize}
\end{itemize}
\begin{table}
\begin{centering}
\begin{tabular}{|c|c|c|}
\hline 
Range of $\tau$ & probability & expected communication\tabularnewline
\hline 
\hline 
$\tau<1/U$ & $1/U$ & $3+2\sqrt{U}$\tabularnewline
\hline 
$1/U<\tau<1/8$ & $1/8-1/U$ & zero\tabularnewline
\hline 
$\tau\in[1/8,1/8+0.03)$ & $0.03$ & $3+O(2^{-\sqrt{U}})$\tabularnewline
\hline 
$\tau>1/8+0.03$ & $7/8-0.03$ & $1$\tabularnewline
\hline 
\end{tabular}
\par\end{centering}
\caption{\label{tab:A-cup-B-1-1}Proof of Theorem \ref{thm:DDT}, protocol
for $\mathcal{A}\cup\mathcal{B}$ in the}
\end{table}

The Buyer's total expected communication is therefore at most
\[
\left(\frac{1}{8}-\frac{1}{U}\right)\cdot0+(\frac{7}{8}-0.03)\cdot1+0.03\cdot\left(3+O(2^{-\sqrt{U}})\right)+\frac{1}{U}\cdot O(\sqrt{U})=0.935+O(\sqrt{1/U})<0.94.
\]

\paragraph{A protocol for the general case}

For the general case, we add a preliminary stage of communication
where the Buyer needs to communicate whether his type is in $\mathcal{Z}$,
$\mathcal{W}$, or $\mathcal{A}\cup\mathcal{B}$. The encoding for
each of those three options is chosen at random as follows:
\begin{itemize}
\item With probability $0.98$, the Buyer should use $00$ when his type
is in $\mathcal{Z}$, and $01$ when his type is in $\mathcal{W}$;
the signal $1$ is reserved for $\mathcal{A}\cup\mathcal{B}$, which
then follows by the above $0.94$-bit protocol for $\mathcal{A}\cup\mathcal{B}$.
\item With probability $0.01$ the Buyer should use $00$ when his type
is in $\mathcal{Z}$, and $01$ when his type is in $\mathcal{A}\cup\mathcal{B}$;
the signal $1$ is reserved for $\mathcal{W}$. (Signal $01$ is followed
by the above $0.94$-bit protocol.)
\item With probability $0.01$ the Buyer should use $00$ when his type
is in $\mathcal{A}\cup\mathcal{B}$, and $01$ when his type is in
$\mathcal{W}$; the signal $1$ is reserved for $\mathcal{Z}$. (Signal
$00$ is followed by the above $0.94$-bit protocol.)
\end{itemize}
For each of $\mathcal{Z},\mathcal{W}$, the Buyer needs to send two
bits with total probability $0.99$, and one bit with probability
$0.01$. Hence his total expected communication is $1.99$ bits. For
$\mathcal{A}\cup\mathcal{B}$, the Buyer sends in expectation $1.02$
bits during the preliminary stage, so his total expected communication
is bounded by $1.96$ bits.
\end{proof}

\section{Approximately ex-post IR protocols\label{sec:Approximately-ex-post-IR}}

In this section we show a generic transformation of IC (and interim
IR) auction protocols to approximately ex-post IR auction protocols.
The main idea is to allow the Buyer, before every node of Chance of
the protocol, to hedge against the risk posed by the randomness of
this node.
\begin{defn}
We say that an auction protocol is {\em ex-post $\varepsilon$-IR}
if it is IC, and whenever the Buyer follows the suggested Buyer's
strategy for his valuation type, at the end of the protocol his payment
is at most $\varepsilon$ greater than his value for the allocation.
\end{defn}

\begin{thm}
Let $\mathcal{\ensuremath{D}}$ be any prior over Buyer's combinatorial
valuations over $n$ items upper bounded by $U$. Given
an IC auction protocol ${\cal P}$, we can transform it into an IC
and ex-post $\varepsilon$-IR auction protocol ${\cal P}^{'}$ with
the same expected payment and allocation; for any random seed $\pi$,
if the communication complexity of ${\cal P}$ is $C$, then the communication
complexity of ${\cal P}^{'}$ with the same random seed is bounded
by $O(C\log(U/\varepsilon)+C^{2})$.
\end{thm}

\begin{proof}
We consider the communication protocol tree associated with ${\cal P}$
and iteratively transform it into ${\cal P}^{'}$ starting from the
root. At the beginning of the protocol, we ask the Buyer to (approximately)
specify his utility (value for the allocation he will receive minus
payment he will be charged) at the end of auction protocol, in expectation
over protocol's randomness and assuming he honestly follows the suggested
Buyer's strategy. Before each node of Chance, we ask the Buyer to
(approximately) specify his expected utility conditioned on each outcome
of the node of Chance. For now we describe and analyze the rest of
the transformation as if the Buyer exactly specifies his values, and
analyze the accumulated error due to finite precision later.

Let $\overline{U}$ denote the Buyer's expected utility at the beginning
of the protocol. Since the protocol is interim IR, $\overline{U}\ge0$.
At each node of Chance we constrain the probability-weighted sum of
expected utilities reported by the Buyer to be exactly $\overline{U}$.
For each outcome $x$ of the node of Chance, let $\overline{U_{x}}$
denote the expected utility reported by the Buyer conditioned on this
outcome. We add $\overline{U_{x}}-\overline{U}$ to the payment in
every leaf of the sub-tree corresponding to $x$. We now make the
following observations about this transformation:
\begin{enumerate}
\item The transformation does not change the Buyer's incentives conditioned
on $x$ since we add the same amount to the payment in every leaf;
in particular, if the sub-protocol conditioned on $x$ was IC before
the transformation, it remains IC after the transformation.
\item The transformation changes the Buyer's expected utility conditioned
on $x$ from $\overline{U_{x}}$ to $\overline{U}$ (here we use that
by the previous observation, the Buyer continues to follow the suggested
Buyer's strategy).
\item In expectation over the randomness in this node of Chance, this transformation
has zero net effect over the Buyer's payment (here use the constraint
that the weighted sum of $\overline{U_{x}}$'s is $\overline{U}$).
In particular, the Buyer has no incentive to misreport his utilities
for the respective outcomes. (In fact, the risk-neutral Buyer is completely
indifferent between any valid report of $\overline{U_{x}}$, so even
when we later constrain the communication to finite precision, he
has no incentive to deviating from reporting the best approximation.)
Furthermore, this guarantees that the Buyer's incentives for arriving
at this node of Chance did not change by the transformation.
\end{enumerate}
Using these observations, it follows by induction that the auction
protocol remains IC after the transformation, and the Buyer always
has ex-post utility $\overline{U}\ge0$.

We now revisit this transformation with restriction to finite precision.
At the $i$-th node of Chance, we ask the Buyer to report $\overline{U_{x}}$'s
to within $\pm\varepsilon2^{-i}$. Then the total error accumulated
on any path of the protocol is always bounded by $\pm\varepsilon$.
The communication complexity is $\sum_{i=1}^{C}\log(U/\varepsilon)+i=O(C\log(U/\varepsilon)+C^{2})$.
\end{proof}

\section{Approximately optimal revenue with finite valuations\label{sec:Approximately-optimal-revenue}}

In this appendix, per the request of a reviewer, we outline a very
short proof of the following proposition: For additive valuations
over independent items, assuming a finite bound on the Buyer's valuation
has an arbitrarily small impact on revenue. Note that \cite{BabaioffGN17}
proved a stronger result with an asymptotic analysis of this finite
bound; however we believe that our proof is simpler. 

We prove the proposition for the case where the optimal revenue is
finite; when the optimal revenue is infinite, the supremum revenue
of trivial mechanisms like only auctioning the grand bundle of all
items is also infinite \cite{HartN17}, so the discussion of menu-size
or communication complexity is less interesting.
\begin{prop}
Let $\mathcal{\ensuremath{D}}$ be any prior over (possibly unbounded)
Buyer's additive valuations over $n$ independent items, and let ${\cal M}$
be any mechanism that obtains finite revenue on ${\cal D}$. For every
$\varepsilon$ there exists $U(\varepsilon)$ such that the same mechanism
(${\cal M}$) obtains $(1-\varepsilon$)-fraction of its revenue if
the Buyer's valuations are capped at $U(\varepsilon)$.
\end{prop}

\begin{proof}
We consider the partitioning of the type-space (aka $\mathbb{R}_{\ge0}^{n}$)
into countable (but infinite) hyperrectangles:
\[
\left(\{0\}\cup\left\{ [2^{i},2^{i+1})\right\} _{i\in\mathbb{Z}}\right)^{n}.
\]
For each hyperrectangle, we consider the contribution to ${\cal M}$'s
revenue of types in this hyperrectangle. We arrange the hyperrectangles
in decreasing order of contributions to revenue (breaking ties arbitrarily).
${\cal M}$'s total revenue is the (countable) sum of contributions
from hyperrectangles in this sequence. Since the sum is finite, it
can be approximated to within $(1-\varepsilon)$-factor by the first
$N(\varepsilon)$ terms in the sequence, for some finite $N(\varepsilon)$.
Hence it suffices to take $U(\varepsilon)$ to be $n\cdot2^{i+1}$
for the maximal $i$ used in any of the hyperrectangles in those first
$N(\varepsilon)$ terms. 
\end{proof}

\end{document}